\documentclass[a4paper, onecolumn, 9pt]{article}
\usepackage[paperwidth=210mm,paperheight=297mm,centering,hmargin=2.3cm,vmargin=2.7cm]{geometry}
\usepackage{marvosym}
\usepackage{amsfonts}
\usepackage{amsthm}
\usepackage{dcolumn}
\usepackage{graphicx}
\usepackage{bm,bbm}
\usepackage{kpfonts}
\usepackage{braket}
\usepackage{color}

\newcommand{\cA}{{\mathcal A}}

\newcommand{\cC}{{\mathcal C}}
\newcommand{\cD}{{\mathcal D}}

\newcommand{\cF}{{\mathcal F}}
\newcommand{\cG}{{\mathcal G}}
\newcommand{\cH}{{\mathcal H}}

\newcommand{\cJ}{{\mathcal J}}
\newcommand{\cK}{{\mathcal K}}
\newcommand{\cL}{{\mathcal L}}
\newcommand{\cM}{{\mathcal M}}
\newcommand{\cN}{{\mathcal N}}

\newcommand{\cP}{{\mathcal P}}

\newcommand{\cR}{{\mathcal R}}
\newcommand{\cS}{{\mathcal S}}
\newcommand{\cT}{{\mathcal T}}
\newcommand{\cU}{{\mathcal U}}

\newcommand{\cX}{{\mathcal X}}
\newcommand{\cY}{{\mathcal Y}}

\newcommand{\fI}{{\mathfrak I}}

\newcommand{\bbmE}{{\mathbbm E}}
\newcommand{\bbmN}{{\mathbbm N}}

\newcommand{\bbmeins}{{\mathbbm 1}}

\newcommand{\id}{\mathrm{id}}

\newcommand{\tr}{\mathrm{tr}}

\newcommand{\cl}{\mathrm{cl}}

\newcommand{\conv}{\mathrm{conv}}

\newtheorem{theorem}{Theorem}

\newtheorem{definition}[theorem]{Definition}

\newtheorem{lemma}[theorem]{Lemma}

\newtheorem{proposition}[theorem]{Proposition}
\newtheorem{remark}[theorem]{Remark}

\newtheorem{fact}[theorem]{Fact}

\begin{document}
	\title{Simultaneous transmission of classical and quantum information under channel uncertainty and jamming attacks}

	\author{Holger Boche, Gisbert Jan\ss en, Sajad Saeedinaeeni\\
		\scriptsize{Electronic addresses: \{boche, gisbert.janssen, sajad.saeedinaeeni\}@tum.de}
		\vspace{0.2cm}\\
		{\footnotesize Lehrstuhl f\"ur Theoretische Informationstechnik, Technische Universit\"at M\"unchen,}\\
		{\footnotesize 80290 M\"unchen, Germany }}
	
	\maketitle
	\section{Abstract}
	We derive universal codes for simultaneous transmission of classical messages and entanglement through quantum channels, possibly under attack of a malignant third party. These codes are robust to different kinds of channel uncertainty. To construct such universal codes, we invoke and generalize properties of random codes for classical and quantum message transmission through quantum channels. We show these codes to be optimal by giving a multi-letter characterization of regions corresponding to capacity of compound quantum channels for simultaneously transmitting and generating entanglement with classical messages. Also, we give dichotomy statements in which we characterize the capacity of arbitrarily varying quantum channels for simultaneous transmission of classical messages and entanglement. These include cases where the malignant jammer present in the arbitrarily varying channel model is classical (chooses channel states of product form) and fully quantum (is capable of general attacks not necessarily of product form). 
	\section{Introduction}
   In real world communication using quantum or classical systems, the parameter determining the channel in use may belong to an uncertainty set, rendering the protocols that assume the channel to be perfectly known practically obsolete. Given such uncertainty, when using the channel many times, as done in Shannon theoretic information processing tasks, assuming the channel to be memoryless or fully stationary is not realistic. In this paper, we consider three models that include channel uncertainty without attempting to reduce it via techniques such as channel identification or tomography. We refer to these models as the compound, arbitrarily varying and fully quantum arbitrarily varying channel models. Each of these models are considered here for transmission of entanglement and classical messages simultaneously between a sender and receiver.\\
   Informally, the first two channel models consist of a set of quantum channels $\{\cN_{s}\}_{s\in S}$ known to the communicating parties. In the compound model, communication is done under the assumption that asymptotically, one of the channels from this set (unknown to the parties) is used in a memoryless fashion. The codes used in this model therefore have to be reliable for the whole family $\{\cN_{s}^{\otimes l}\}_{s\in S}$ of memoryless channels for large enough values of $l\in\mathbb{N}$.\\
    In the arbitrarily varying model, given a number of channel uses $l$, an adversarial party chooses the sequence $s^{l}=(s_{1},\dots,s_{l})\in S^{l}$ unknown to the communication parties, to yield the channel $\cN_{s^{l}}=\bigotimes_{i=1}^{l}\cN_{s_{i}}$. The adversary may choose this sequence knowing the encoding procedure used by the sender. The code in use therefore has to be reliable for the whole family $\{\cN_{s^{l}}\}_{s^{l}\in S^{l}}$ of memoryless channels. Finally, in the third channel model, namely that of the fully quantum arbitrarily varying, the assumption of memoryless communication is dropped. Here, the adversary may choose channel states that are not necessarily of the product form mentioned in the previous model. \\
   The quantum channel has different capacities for information transmission. One may consider the capacity of the channel for public (\cite{holevo, p2}) or private (\cite{idev, caietal}) classical message transmission, entanglement transmission or entanglement generation (\cite{idev}) to name a few. These communication scenarios have been considered subsequently under channel uncertainty (\cite{boche17, mosonyi, c1, c2, ahls, advnoise}). Simultaneous transmission of classical and quantum messages, the subject of this work, has also been of interest(\cite{devetak}). This includes scenarios where the communication parties would like to enhance their classical message transmission by sharing quantum information primarily at their disposal or vice versa(\cite{e1, e2, e3}). The body of research in this area is clearly interesting, when regions beyond those achieved by simple time-sharing between established classical message and quantum information transmission codes are reached. \\
   Simultaneous transmission of classical messages and entanglement is a nontrivial problem even if capacity achieving codes for the corresponding univariate transmission goals are at hand. It was already observed in \cite{devetak} for perfectly known quantum channels that the naive time sharing strategy is generally insufficient to achieve the full capacity region. Examples of channels where coding beyond time-sharing is indispensable does not depend on constructing pathologies. They are readily found even within the standard arsenal of qubit quantum channels, e.g. the dephasing qubit channels \cite{devetak}.\\
    We derive codes for simultaneous transmission of classical messages and entanglement that are robust to the three types of uncertainty mentioned above. The codes used here for the compound model, are different from those used for the point to point communication in \cite{devetak} when considering the special case of $|S|=1$. Given that the input state  approximation techniques used therein prove insufficient in presence of channel state uncertainty, in the present work we use the decoupling approach first established in \cite{klesse}. We combine robust random codes for classical message transmission from \cite{mosonyi} and a generalization of (decoupling based) entanglement transmission codes from \cite{boche17} to construct appropriate simultaneous codes for compound quantum channels under the maximal error criterion. We show that these codes are optimal by giving a multi-letter characterization of the capacity of compound quantum channels with no assumption on, the size of the underlying uncertainty set. We use the asymptotic equivalence of the two tasks of entanglement transmission and entanglement generation to include the capacity region corresponding to simultaneous transmission of classical messages and generation of entanglement between the two parties. \\
	Next, we convert the codes derived for the compound channel, using Ahlswede's robustification and elimination techniques (\cite{ahls}) to derive suitable codes for arbitrarily varying quantum channels. This is possible given that the error functions associated with codes corresponding to the compound model decay to zero exponentially. We derive a dichotomy statement (\cite{ahls}), for the simultaneous classical message and entanglement transmission through AVQCs under the average error criterion. This dichotomy is observed when considering two scenarios where the communicating parties do and do not have access to unlimited common randomness, yielding the common-randomness and deterministic capacity regions of the channel model respectively. Therefore, we show that firstly, the common-randomness capacity region of the arbitrarily varying channel is equal to that of the compound channel $\conv(\cJ)$, namely the compound channel generated by the convex hull of the uncertainty set of channels $\cJ$. Secondly, if the deterministic capacity of the arbitrarily varying channel is not the point $(0,0)$, it is equal to the common-randomness capacity of the channel. \\
	We give a necessary and sufficient condition for the deterministic capacity region to be be the point $(0,0)$. This condition is known as symmetrizablity of the channel (see \cite{advnoise} and \cite{bochenozel}). Finally, we show that the codes derived here, can be used for fully quantum AVCs where the jammer is not restricted to product states, but can use general quantum states to parametrize the channel used many times. This model has been introduced in Section \ref{classicallyenhanced} along with the main result and related work for fully quantum AVCs and hence here, we avoid further explanation of the techniques used there.\\
	The task of simultaneous transmission of classical messages and entanglement was first considered by Devetak and Shor in \cite{devetak} in case of a memoryless quantum channel under assumption that the channels state is perfectly known to its users. The authors derived a multi-letter characterization of the capacity region in this setting which also classified the na\"ive time-sharing approach as being suboptimal for simultaneous transmission. A code construction sufficient to achieve also the rate pairs lying outside the time-sharing region was derived using a "piggy-backing" technique. A specialized construction introduced in \cite{idev} allows to encode the identity of the classical message into the coding states of an underlying entanglement transmission code. The mentioned strategy to optimally combine different communication tasks in quantum channel coding was afterwards used and further developed in different directions. We explicitly mention subsequent research activity by Hsieh and Wilde \cite{e2,e3} where the idea of "piggy backing" classical messages onto quantum codes was extended to include entanglement assistance. The resulting code construction being sufficient to achieve each point in the three-dimensional rate region for entanglement-assisted classical/quantum simultaneous transmission leads to a full
	(multi-letter) characterization of the "Quantum dynamic capacity" of a (perfectly known) quantum channel \cite{hsieh11} (see the textbook \cite{wilde13} for an up-to-date pedagocial presentation of the mentioned results). \newline In order to derive classically enhanced quantum codes being robust against channel uncertainty, we refine the construction entanglement transmission codes for compound quantum channels from \cite{boche17,boche18} instead of elaborating on the usual approach building up on codes from \cite{idev}. In fact, it was noticed earlier that deriving entanglement generation codes from secure classical message transmission codes (the strategy which the arguments in \cite{idev} follow) seems to be not suitable when the channel is a compound quantum channel.\\
	In the first section following this introduction, we introduce the notation used in this work. Precise definitions of the channel models, codes used in different scenarios along with capacity regions and finally the main results in form of Theorem \ref{mainresult} and Theorem \ref{avcqmainresult}, are given in Section \ref{basicdef}. In Section \ref{preliminaryresults}, we present preliminary coding results for entanglement transmission (Section \ref{etcsection}) and classical message transmission (Section \ref{ctresult}). The entanglement transmission codes introduced in this section are a generalization of the random codes in \cite{boche17} and \cite{boche18} to accommodate conditional typicality of the input on words from many copies of an alphabet. The classical message transmission codes are those from \cite{mosonyi} that prove sufficient for our simultaneous coding purposes. \\
	Equipped with these results, we move on to Section \ref{proofsforcompound}, to prove the coding results for the compound channel model. In this section, after proving a converse for the capacity region in Theorem \ref{mainresult}, we prove the direct part in two steps. In the first step, we show that capacity regions that correspond to the case where the sender is restricted to inputting maximally entangled pure states are achieved. In the second step, we prove achievablity of capacity regions corresponding to general inputs, using elementary methods that are less involved that the usual BSST type results used for this generalization in \cite{boche17} and \cite{boche18}. \\
	In Section \ref{avqcsection}, after proving a converse for the capacity region under the arbitrarily varying channel model, we prove coding results in this model by converting the compound channel model codes using Ahlswede's robustification method. This, assumes unlimited common randomness available to the legal parties. We then use an instance of elimination to show that if the deterministic capacity region is not the point $(0,0)$, negligible amount of common randomness per use of the channel is sufficient to achieve the same capacity region. Also in this section, we prove necessity and sufficiency of symmetrizablity condition for the case where the deterministic capacity region is the point $(0,0)$. Finally, in Section \ref{classicallyenhanced}, we generalize these results to the case of quantum jammer by proving Theorem \ref{thm:full_avqc_rand_cap}. \\
	\section{Notations and conventions}
	All Hilbert spaces are assumed to have finite dimensions and are over the field $\mathbb{C}$. All alphabets are also assumed to have finite dimensions. We denote the set of states by $\cS(\cH):=\{\rho\in\cL(\cH):\rho\geq 0,\text{tr}(\rho)=1\}$. Pure states are given by projections onto one-dimensional subspaces. To each subspace $\cF\subset\cH$, we can associate unique projection $q_{\cF}$
	whose range is the subspace $\cF$ and we write $\pi_{\cF}$ for the maximally mixed state on $\cF$, i.e. 
	\begin{equation*}
	\pi_{\cF} := \frac{q_{\cF}}{\text{tr}(q_{F} )}.
	\end{equation*}
	The set of completely positive trace preserving (CPTP) maps between the operator spaces $\cL(\cH_{A})$ and $\cL(\cH_{B})$
	is denoted by $\cC(\cH_{A}, \cH_{B})$. Thus $\cH_{A}$, plays the role of the input Hilbert space to the channel (traditionally
	owned by Alice) and $\cH_{B}$ is channel's output Hilbert space (usually in Bob's possession). $\cC^{\downarrow}(\cH_{A}, \cH_{B})$ stands for the set of completely positive trace decreasing maps between $\cL(\cH_{A})$ and $\cL(\cH_{B})$. $\cU(\cH)$ will denote in what follows, the group of unitary operators acting on $\cH$. For a Hilbert space $\cG\subset\cH$, we will always identify $\cU(\cG)$ with a subgroup of $\cU(\cH)$. For any projection $q\in\cL(\cH)$ we set $q^{\perp}:= 1_{\cH}-q$.\\
	Each projection $q \in \cL(\cH)$ defines a completely positive trace decreasing map $Q$ given by $Q(a):= qaq$ for
	all $a\in\cL(\cH)$. In a similar fashion, any $U \in \cU(\cH)$ defines a $\cU \in\cC(\cH,\cH)$ by $\cU(a) := UaU^{\dagger}$ for $a\in\cL(\cH)$. 
	We use the base two logarithm which is denoted by $\log$. The von Neumann entropy of a state $\rho \in \cS(\cH)$ is given by
	\begin{equation*}
	S(\rho):=-\text{tr}(\rho\log\rho).
	\end{equation*}
	The coherent information for $\cN \in \cC(\cH_{A}, \cH_{B})$ and $\rho \in \cS(\cH_{A})$ is defined by
	\begin{equation*}
	I_{c}(\rho, \cN ) := S(\cN (\rho)) -S((\text{id}_{\cH_{A}} \otimes\cN )(\ket{\psi}\bra{\psi}))
	\end{equation*}
	where $\psi\in\cH_{A} \otimes\cH_{A}$ is an arbitrary purification of the state $\rho$. We also use $S_{e}(\rho,\cN):=S((\text{id}_{\cH_{A}} \otimes\cN )(\ket{\psi}\bra{\psi}))$ to denote entropy exchange. A useful equivalent definition of $I_{c}(\rho, \cN )$ is given in terms of $\cN \in \cC(\cH_{A}, \cH_{B})$ and any complementary channel $\hat{\cN } \in \cC(\cH_{A}, \cH_{e})$ where $\cH_{e}$ denotes the Hilbert space of the environment. Due to Stinespring's dilation theorem, $\cN $ can be represented as
	\begin{equation*}
	\cN (\rho) = \text{tr}_{\cH_{e}}(v\rho v^{*})
	\end{equation*}
	for $\rho \in \cS(\cH_{A})$ where $v :\cH_{A} \to \cH_{B} \otimes\cH_{e}$ is a linear isometry. The complementary channel $\hat{\cN}  \in \cC(\cH_{A}, \cH_{e})$ of $\cN$ is given by
	\begin{equation*}
	\hat{\cN }(\rho) := \text{tr}_{\cH_{B}}(v\rho v^{*}) .
	\end{equation*}
	The coherent information can then be written as
	\begin{equation}\label{def2cinf}
	I_{c}(\rho,\cN)=S(\cN(\rho))-S(\hat{\cN}(\rho)).
	\end{equation}
	This quantity can also be defined in terms of the bipartite state $\sigma\in\cS(\cH_{A}\otimes\cH_{B})$ with 
	\begin{equation*}
	\sigma:=\text{id}_{\cH_{A}}\otimes\cN(\ket{\psi}\bra{\psi})
	\end{equation*}
	as
	\begin{equation*}
	I(A\rangle B,\sigma):=S(\sigma^{B})-S(\sigma)
	\end{equation*}
	where $\sigma^{B}$ is the marginal state given by $\sigma^{B}:=\tr_{A}(\sigma)$ and we have the identity
	\begin{equation*}
	I_{c}(\rho,\cN)=I(A\rangle B,\sigma).
	\end{equation*}
	As a measure of closeness between two states $\rho, \sigma \in \cS(\cH)$, we may use the fidelity $F(\rho, \sigma) := \parallel\sqrt{\rho}
	\sqrt{\sigma}\parallel_{1}^{2}$. The
	fidelity is symmetric in the input and for a pure state $\rho = \ket{\phi}\bra{\phi}$, we have $F(\ket{\phi}\bra{\phi}, \sigma)=\braket{\phi, \sigma\phi}$.
	A closely related quantity is the entanglement fidelity, which for $\rho \in \cS(\cH_{A})$ and $\cN\in\cC^{\downarrow}(\cH_{A}, \cH_{B})$, is given by
	\begin{equation*}
	F_{e}(\rho, \cN ) := \braket{\psi,(\text{id}_{\cH} \otimes \cN )(\ket{\psi}\bra{\psi})\psi}
	\end{equation*}
	with $\psi \in \cH_{A}\otimes\cH_{A}$ an arbitrary purification of the state $\rho$. 
	\\
	Another quantity that will be significant in the present work is the quantum mutual information (see e.g \cite{wilde13}). For a state $\rho\in\cS(X\otimes\cH_{B})$, the quantum mutual information is defined as
	\begin{equation*}
	I(X;B,\rho):=S(\rho^{X})+S(\rho^{B})-S(\rho)
	\end{equation*}
	where $\rho^{X}$ and $\rho^{B}$ are marginal states of $\rho$. \\
	For the approximation of arbitrary compound channels (introduced in the next section) by finite ones we use the diamond norm $\parallel\cdot\parallel_{\diamond}$, given for any $\cN:\cL(\cH_{A})\to\cL(\cH_{B})$ by
	\begin{equation*}
	\parallel\cN\parallel_{\diamond} := \sup_
	{n\in\mathbb{N}}
	\max_{
		a\in\cL(\mathbb{C}^{n}\otimes\cH),\parallel a\parallel_{1}=1}
	\parallel(\text{id}_{n} \otimes \cN )(a)\parallel_{1},
	\end{equation*}
	where $\text{id}_{n}: \cL(\mathbb{C}^{n}) \to \cL(\mathbb{C}^{n})$ is the identity channel. We state the following facts about $|| \cdot ||_{\diamond}$ (see e.g \cite{dnomr}). First, $||\cN ||_{\diamond} = 1$ for all $\cN \in\cC(\cH_{A}, \cH_{B})$. Thus, $\cC(\cH_{A}, \cH_{B})\subset S_{\diamond}$, where $S_{\diamond}$ denotes the unit sphere of the normed
	space $(\cL(\cH_{A}), \cL(\cH_{B}), || \cdot ||_{\diamond})$. Moreover, $||\cN_{1} \otimes \cN_{2}||_{\diamond} = ||\cN_{1}||_{\diamond}||\cN_{2}||_{\diamond}$ for arbitrary linear maps $\cN_{1}, \cN_{2} :\cL(\cH_{A}) \to \cL(\cH_{B})$. Throughout this work we have made use of the idea of nets to approximate arbitrary compound quantum channels using ones with finite uncertainty sets. This idea is presented in Appendix \ref{mosonyicodes} by Definition \ref{netdef} and proceeding two lemmas.\\
	 We use $\epsilon_{n} \to 0$ exponentially as $n\to\infty$  or we say $\epsilon_{n}$ approaches (goes to) zero exponentially, if $-\frac{1}{n}\log\epsilon_{n}$ is a strictly positive constant. For $\epsilon_{1,n}$ and $\epsilon_{2,n}$ both approaching zero exponentially, we use $\epsilon_{1,n}\geq\epsilon_{2,n}$ if $-\frac{1}{n}\log\epsilon_{1,n}\leq-\frac{1}{n}\log\epsilon_{2,n}$. We use $\cl(A)$ to denote the closure of set $A$ and finally, we use $\mathfrak{S}_{n}$ to denote the group of permutations on $n$ elements such that $\alpha(s^{n})=(s_{\alpha(1)},\dots,s_{\alpha(n)})$ for each $\alpha\in\mathfrak{S}_{n}$ and $s^{n}=(s_{1},\dots,s_{n})\in S^{n}$. 
	
	\section{Basic definitions and main results}\label{basicdef}
	We consider two channel models of compound and arbitrarily varying quantum channels. They are both generated by an uncertainty set of CPTP maps. For the purposes of the present work, when considering the arbitrarily varying channel model, we assume finiteness of the generating uncertainty set. This assumption is absent in the case of the compound channel model. 
	\subsection{The compound quantum channel}\label{basicdefcompound}
	Here, we consider quantum compound channels. Let  $\cJ:=\{\cN_{s}\}_{s\in S}\subset\cC(\cH_{A},\cH_{B})$ be a set of CPTP maps. The compound quantum channel generated by $\cJ$ is given by family  $\{\cN^{\otimes n}:\cN\in \cJ\}_{n=1}^{\infty}$. 
	In other words, using $n$ instances of the compound channel is equivalent to using $n$ instances of one of the channels from the uncertainty set. The users of this channel may or may not have access to the Channel State Information (CSI). We will often use the set $S$ to index members of $\cJ$. A compound channel is used $n\in\mathbb{N}$ times by the sender Alice, to convey classical messages from a set $[M_{1,n}]:=\{1,...,M_{1,n}\}$ to a receiver Bob. At the same time, the parties would like to communicate quantum information. Here, we consider two scenarios in which quantum information can be communicated between the parties. \\
	\textbf{Classically Enhanced Entanglement Transmission (CET)}: While transmitting classical messages using $n\in\mathbb{N}$ instances of the compound channel, the sender wishes to transmit the maximally entangled state in her control to the receiver. The subspace $\cF_{A,n}$ with $\cF_{A,n}\subset\cH_{A}^{\otimes n}$ and $M_{2,n}:=\dim(\cF_{A,n})$, quantifies the amount of quantum information transmitted. More precisely:
	\begin{definition}\label{scodes}
		An $(n,M_{1,n},M_{2,n})$ CET code for $\mathcal{J}\subset\cC(\cH_{A},\cH_{B})$, is a family $\mathcal{C}_{CET}:=(\mathcal{P}_{m},\mathcal{R}_{m})_{m\in[M_{1,n}]}$ with
		\begin{itemize}
			\item{$\mathcal{P}_{m}\in \cC(\mathcal{F}_{A,n},\mathcal{H}_{A}^{\otimes n})$},
			\item{ $\mathcal{R}_{m}\in \cC^{\downarrow}(\mathcal{H}_{A}^{\otimes n},\mathcal{F}_{B,n})$} with $\cF_{A,n}\subset\cF_{B,n}$ and
			\item $\sum_{m\in[M_{1,n}]}\cR_{m}\in\cC(\cH_{B}^{\otimes n},\cF_{B,n})$.
		\end{itemize}
		
	\end{definition}
	\begin{remark} 
		We remark that as defined above, for each $m\in[M_{1,n}]$ we have a $(n,M_{2,n})$ entanglement transmission code for $\cJ$.
	\end{remark}
	
	For every $m\in M_{1,n}$ and $s\in S$, we define the following performance function for this communication scenario when $n\in\mathbb{N}$ instances of the channel have been used,
	\begin{equation*}
P(\cC_{CET},\cN_{s}^{\otimes n},m):=	F(\ket{m}\bra{m}\otimes \Phi^{AB},\id_{\cF_{A,n}}\otimes\cR\circ\cN_{s}^{\otimes n}\circ\cP_{m}(\Phi^{AA})),
	\end{equation*}
	where $\Phi^{XY}$ is a maximally entangled state on $\cF_{X,n}\otimes\cF_{Y,n}$ and 
	\begin{equation*}
	\cR:=\sum_{m\in[M_{1,n}]}\ket{m}\bra{m}\otimes\cR_{m}.
	\end{equation*}
	\textbf{Classically Enhanced Entanglement Generation (CEG)}:
	In this scenario, while transmitting classical messages, Alice wishes to establish a pure state shared between her and Bob. As the maximally entangled pure state shared between the parties is an instance of such a pure state, it can be proven that the previous task achieved in CET, achieves the task laid out by this one, but the opposite is not necessarily true. More precisely:
	\begin{definition}\label{sscodes}
		An $(n,M_{1,n},M_{2,n})$ CEG code for $\mathcal{J}\subset\cC(\cH_{A},\cH_{B})$, is a family $\cC_{CEG}:=(\Psi_{m},\mathcal{R}_{m})_{m=1}^{M_{1,n}}$, where $\Psi_{m}$ is a pure state on $\mathcal{F}_{A,n}\otimes\mathcal{H}_{A}^{\otimes n}$ and
		\begin{itemize}
			\item $\mathcal{R}_{m}\in \cC^{\downarrow}(\mathcal{H}_{B}^{\otimes n},\mathcal{F}_{B,n})$ with $\cF_{A,n}\subset\cF_{B,n}$ and
			\item $\sum_{m\in[M_{1,n}]}\cR_{m}\in\cC(\cH_{B}^{\otimes n},\cF_{B,n})$.
		\end{itemize}
	\end{definition}
	The relevant performance functions for this task, for every $m\in [M_{1,n}]$ and $s\in S$, are
	\begin{equation}
		P(\cC_{CEG},\cN_{s}^{\otimes n},m):=F(\ket{m}\bra{m}\otimes \Phi,\id_{\cF_{A,n}}\otimes\cR\circ\cN_{s}^{\otimes n}(\Psi_{m})),
	\end{equation}
	with $\Phi$ maximally entangled on $\mathcal{F}_{A,n}\otimes\mathcal{F}_{B,n}$.\\
	Averaging over the message set $[M_{1,n}]$, will give us the corresponding average performance functions for each $s\in S$,
	\begin{equation*}
\overline{P}(\cC_{X},\cN_{s}^{\otimes n}):=\frac{1}{M_{1,n}}\sum_{m\in [M_{1,n}]}	P(\cC_{X},\cN_{s}^{\otimes n},m),
	\end{equation*}
	for $X\in\{CET,CEG\}$. For each scenario, we define the achievable rates.
	\begin{definition}
		Let $X\in\{CET, CEG\}.$ A pair $(R_{1},R_{2})$ of non-negative numbers is called an achievable
		X rate for the compound channel $\cJ$, if for each $\epsilon,\delta>0$ exists a number $n_{0}=n_{0}(\epsilon,\delta)$, such that for each
		$n > n_{0}$ we find and $(n,M_{1,n},M_{2,n})$ X code $\cC_{X}$ such that
		\begin{enumerate}
			\item $\frac{1}{n}\log M_{i,n}\geq R_{i}-\delta$ for $i\in\{1,2\}$,
			\item $\inf_{s\in S}\min_{m\in M_{1,n}} P(\cC_{X},\cN_{s}^{\otimes n},m)\geq 1-\epsilon$
		\end{enumerate}
		are simultaneously fulfilled. We also define X "average-error-rates" by averaging the performance functions in the last condition over $m\in [M_{1,n}]$. We define the X capacity region of $\cJ$ by
		\begin{equation}\label{capacityregions}
		C_{X}(\cJ):=\{(R_{1},R_{2})\in\mathbb{R}^{+}_{0}\times\mathbb{R}^{+}_{0}:(R_{1},R_{2})\text{ is achievable X rate for }\cJ\}.
		\end{equation}
		Also the capacity region corresponding to average error criteria is defined as
		\begin{equation}\label{averageregions}
		\overline{C}_{X}(\cJ):=\{(R_{1},R_{2})\in\mathbb{R}^{+}_{0}\times\mathbb{R}^{+}_{0}:(R_{1},R_{2})\text{ is achievable X average-error-rate for }\cJ\}.
		\end{equation}
	\end{definition}

	Moreover, let $\mathcal{X}$ be an alphabet, $\mathcal{M}\in \cC(\mathcal{H}_{A},\mathcal{H}_{B})$ $\forall s\in S$, $p\in\mathcal{P}(\mathcal{X})$ and $\Psi_{x}$ be a pure state for all $x\in\cX$. Given the state
	\begin{equation}\label{evaluationstate}
	\omega(\mathcal{M},p,\Psi):=\sum_{x\in\cX}p(x)\ket{x}\bra{x}\otimes\id_{\cH_{A}}\otimes\cM(\Psi_{x}),
	\end{equation}
	we introduce the following set,
	\begin{equation*}
	\hat{C}(\mathcal{N}_{s},p,\Psi):=\{(R_{1},R_{2})\in\mathbb{R}^{+}_{0}\times\mathbb{R}^{+}_{0}:R_{1}\leq I(X;B,\omega(\mathcal{N}_{s},p,\Psi))\wedge
	R_{2}\leq I(A\rangle BX,\omega(\mathcal{N}_{s},p,\Psi))\}
	\end{equation*}
	with $\Psi$ denoting $(\Psi_{x}:x\in\cX)$ collectively. We will also use
	\begin{equation*}
	\frac{1}{l}A:=\{(\frac{1}{l}x_{1},\frac{1}{l}x_{2}):(x_{1},x_{2})\in A\}.
	\end{equation*}
	The following statement is the first main result of this paper. 
	\begin{theorem}\label{mainresult}
		Let $\mathcal{J}:=\{\cN_{s}\}_{s\in S}\subset \cC(\mathcal{H}_{A},\mathcal{H}_{B})$ be any compound quantum channel. Then
		\begin{equation*}
		C_{CET}(\mathcal{J})=\overline{C}_{CET}(\mathcal{J})= C_{CEG}(\mathcal{J})=\overline{C}_{CEG}(\mathcal{J})=\cl\bigg(\bigcup_{l=1}^{\infty}\frac{1}{l}\bigcup_{p,\Psi}\bigcap_{s\in S}\hat{C}(\mathcal{N}_{s}^{\otimes l},p,\Psi)\bigg)
		\end{equation*}
		holds.
	\end{theorem}
	This theorem is proven in the following steps. In Section \ref{conversesection}, we prove that  $\overline{C}_{CEG}(\mathcal{J})$ is a subset of the set on the rightmost set in the above equalities. In Section \ref{directsection}, we prove that the rightmost set is a subset of $C_{CET}(\mathcal{J})$. Together with the operational inclusions
	\begin{equation*}
C_{CET}(\mathcal{J})\subset C_{CEG}(\mathcal{J})
	\end{equation*}
	and
	\begin{equation*}
C_{X}(\mathcal{J})\subset \overline{C}_{X}(\mathcal{J})
	\end{equation*}
	for $X\in\{CEG,CET\}$, we conclude the equalities in the statement of the theorem. 
	\subsection{The arbitrarily varying quantum channel}\label{codedefinitionavcq}
	The arbitrarily varying quantum channel generated by a set $\cJ:=\{\cN_{s}\}_{s\in S}$ of CPTP maps with input Hilbert space $\cH_{A}$ and output Hilbert space $\cH_{B}$, is given by family of CPTP maps $\{\cN_{s^{l}}:\cL(\cH_{A}^{\otimes l})\to\cL(\cH_{B}^{\otimes l}),s^{l}\in S^{l}, l\in\mathbb{N}\}_{l=1}^{\infty}$, where 
	\begin{equation*}
	\cN_{s^{l}}:=\cN_{s_{1}}\otimes\dots\cN_{s_{l}}\ \ (s^{l}\in S^{l}).
	\end{equation*}
	We use $\cJ$ to denote the AVQC generated by $\cJ$. To avoid further technicalities, we always assume $|S|<\infty$ for the AVQC generating sets appearing in this paper. Most of the results in this paper may be generalized to the case of general sets by clever use of approximation techniques from convex analysis together with continuity properties of the entropic quantities which appear in the capacity characterizations (see \cite{advnoise}).
	\begin{definition}
		An $(l,M_{1,l},M_{2,l})$ random CET code for $\cJ$ is a probability measure $\mu_{l}$ on $(\cC(\cF_{A,l},\cH_{A}^{\otimes l})^{M_{1,l}}\times\Omega_{l},\sigma_{l})$, where 
		\begin{itemize}
			\item $\Omega_{l}:=\{(\cR^{(1)},\dots,\cR^{(M_{1,l})}),\sum_{m\in[M_{1,l}]}\cR^{(m)}\in\cC(\cH_{B}^{\otimes l},\cF_{B,l})\}$,
			\item $\dim(\cF_{A,l})=M_{2,l},\cF_{X,l}\subset\cH_{X}^{\otimes l}$,$\ \ \ (X\in\{A,B\})$.
			
			\item The sigma-algebra $\sigma_{l}$ is chosen such that the function
			\begin{equation}\label{fidfunc}
			g_{s^{l}}(\cP^{(m)},\cR^{(m)}):=F(\ket{m}\bra{m}\otimes \Phi^{AB},\id_{\cH_{A}^{\otimes l}}\otimes\cR\circ\cN_{s^{l}}\circ\cP^{(m)}(\Phi^{AA}))
			\end{equation}
			is measurable with respect to $\mu_{l}$, for all $m\in[M_{1,l}],s^{l}\in S^{l}$. In (\ref{fidfunc}), $\Phi^{XY}$ is a maximally entangled state on $\cF_{X,l}\otimes\cF_{Y,l}$ and $\cR:=\sum_{m\in[M_{1,l}]}\ket{m}\bra{m}\otimes \cR^{(m)}$. 
			\item We further require that $\sigma_{l}$ contains all the singleton sets. The case where $\mu_{l}$ is deterministic, namely is equal to unity on a singleton set and zero otherwise, gives us a deterministic $(l,M_{1,l},M_{2,l})$ CET codes for $\cJ$. Abusing the terminology, we also refer to the singleton sets as deterministic codes.
		\end{itemize}
	\end{definition}
	\begin{definition}
		A non-negative pair of real numbers $(R_{1},R_{2})$ is called an achievable CET rate pair for $\cJ:=\{\cN_{s}\}_{s\in S}$ with random codes and average error criterion, if there exists a random CET code $\mu_{l}$ for $\cJ$ with members of singleton sets notified by $(\cP^{(m)},\cR^{(m)})_{m\in [M_{1,l}]}$ such that 
		\begin{enumerate}
			\item $\liminf_{l\to\infty}\frac{1}{l}\log M_{i,l}\geq R_{i}\ (i\in\{1,2\})$,
			\item $\lim_{l\to\infty}\inf_{s^{l}\in S^{l}}\int\frac{1}{M_{1,l}}\sum_{m\in[M_{1,l}]}g_{s^{l}}(\cP^{(m)},\cR^{(m)}) \ d\mu_{l}(\cP^{(m)},\cR^{(m)})_{m=1}^{M_{1,l}}=1.$
		\end{enumerate}
	\end{definition}
	The random CET capacity region with average error criterion of $\cJ$ is defined by
	\begin{align*}
	\overline{\cA}_{r,CET}&(\cJ):=\{(R_{1},R_{2}): (R_{1},R_{2})\ is \ achievable \ CET\  rate\  pair\  for \ \cJ\\& with\  random\  codes\  and \ average\  error\  criterion\}.
	\end{align*}
	\begin{definition}
		A non-negative pair of real numbers $(R_{1},R_{2})$ is called an achievable deterministic CET rate for $\cJ$ with average error criterion, if there exists a deterministic $(l,M_{1,l},M_{2,l})$ CET code $(\cP^{(m)},\cR^{(m)})_{m\in[M_{1,l}]}$ for $\cJ$ with
		\begin{enumerate}
			\item $\liminf_{l\to\infty}\frac{1}{l}\log M_{i,l}\geq R_{i}\ (i\in\{1,2\})$,
			\item $\lim_{l\to\infty}\inf_{s^{l}\in S^{l}}\frac{1}{M_{1,l}}\sum_{m\in[M_{1,l}]}g_{s^{l}}(\cP^{(m)},\cR^{(m)})=1$
		\end{enumerate}
	\end{definition}
	Correspondingly we define the following capacity region,
	\begin{align*}
	\overline{\cA}_{d,CET}&(\cJ):=\{(R_{1},R_{2}): (R_{1},R_{2})\ is \ achievable \ deterministic\\&
	 \ CET\  rate\  pair\  for \ \cJ\ with\ average\  error\  criterion\}.
	\end{align*}
	The deterministic CET codes defined here, are entanglement transmission codes for each $m\in[M_{1,l}]$. More precisely we have the following definition.
	\begin{definition}
		An $(n,M)$, $n,M\in\mathbb{N}$, entanglement transmission code for AVQC $\cJ\subset\cC(\cH_{A},\cH_{B})$ is a pair $(\cP,\cR)$ with $\cP\in\cC(\cF_{A,n},\cH_{A}^{\otimes n}),\cR\in\cC(\cH_{B}^{\otimes n},\cF_{B,n})$ with $\cF_{A,n}\subset\cF_{B,n}\subset\cH_{A}^{\otimes n}$ and $\dim(\cF_{A,n})=M$. The corresponding performance function for this task is 
		\begin{equation*}
		F(\Phi^{AB},\id_{\cH_{A}^{\otimes n}}\otimes\cR\circ\cN_{s^{n}}\circ\cP(\Phi^{AA})),\ \ s^{n}\in S^{n}.
		\end{equation*}
	\end{definition}
Essential to the statement of our results is the concept of symmetrizablity defined in the following.
\begin{definition}
Let $\cJ:=\{\cN_{s}\}_{s\in S}\subset\cC(\cH_{A},\cH_{B})$ with $|S|<\infty$ be an AVQC. 
\begin{enumerate}
	\item $\cJ$ is called $l$-symmetrizable for $l\in\mathbb{N}$, if for each finite set $\{\rho_{1},\dots,\rho_{K}\}\subset\cS(\cH_{A}^{\otimes l})$ with $K\in\mathbb{N}$, there is a map $p:\{\rho_{1},\dots,\rho_{K}\}\to\cP(S^{l})$ such that for all $i,j\in\{1,\dots,K\}$
	\begin{equation}\label{condsymm}
	\sum_{s^{l}\in S^{l}}p(\rho_{i})(s^{l})\cN_{s^{l}}(\rho_{j})=\sum_{s^{l}\in S^{l}}p(\rho_{j})(s^{l})\cN_{s^{l}}(\rho_{i}).
	\end{equation}
	\item We call $\cJ$ symmetrizable if it is $l$-symmetrizable for all $l\in\mathbb{N}$. 
\end{enumerate}
\end{definition}
\begin{remark}
	The above definition for symmetrizablity was first established in \cite{advnoise}, generalizing the concept of symmetrization for classical AVQCs from \cite{Ericson}. This definition for symmetrizablity was meaningfully simplified in \cite{bochenozel}, to require checking of the condition (\ref{condsymm}) for two input states only (K=2).
\end{remark}  
We prove the following result to be the second main result of this paper.
\begin{theorem}\label{avcqmainresult}
	Let $\cJ:=\{\cN_{s}\}_{s\in S}\subset\cC(\cH_{A},\cH_{B})$ with $|S|<\infty$ be an AVQC. The following hold.
	\begin{enumerate}
		\item $\overline{\cA}_{d,CET}(\cJ)\neq\{(0,0)\}$ implies 
	\begin{equation}
	\overline{\cA}_{d,CET}(\cJ)=\overline{\cA}_{r,CET}(\cJ)=\overline{C}_{CET}(\conv(\cJ)),
	\end{equation}
	where $\overline{C}_{CET}(\cM)$ is the CET capacity of compound channel $\cM$ with average error criterion defined in the previous section and 
	\begin{equation*}
	\conv(\cJ):=\{\cN_{q}:\cN_{q}:=\sum_{s\in S}q(s)\cN_{s},q\in\cP(S)\}.
	\end{equation*}
	\item $\overline{\cA}_{d,CET}(\cJ)=\{(0,0)\}$ if and only if $\cJ$ is symmetrizable.
	\end{enumerate}
\end{theorem}
\section{Universal random codes for quantum channels}\label{preliminaryresults}
In this section we prove universal random coding results for entanglement transmission
and classical message transmission over quantum channels. Most of the statements below, are
implicitly contained in the literature. We state some properties of these codes that stem from their random nature and prove useful when deriving CET codes stated in Section \ref{proofsforcompound}. \\
Before proceeding with the following two sections in which we introduce appropriate entanglement transmission and classical message transmission coding results and for the reader's convenience, we present briefly the concept of types used in the remainder of this section. For more information on the concept of types, see e.g. \cite{wilde13}.\\
For $l\in\mathbb{N}$, the word $x^{l}\in\cX^{l}$ that is a string of letters $x\in\cX$ and the state $\rho$ with spectral decomposition $\rho:=\sum_{x\in\mathcal{X}}p(x)\ket{x}\bra{x}$, we define the $\delta$-typical (frequency typical) projection
\begin{equation*}
q_{\delta,l}(\rho):=\sum_{x^{l}\in T_{p,\delta}^{l}}\ket{x^{l}}\bra{x^{l}},
\end{equation*}
where $T_{p,\delta}^{l}$ is the set of $\delta$-typical sequences in $\mathcal{X}^{l}$, defined by
\begin{equation}\label{defoftypset}
T_{p,\delta}^{l}:=\{x^{l}:\forall x\in\cX,|\frac{1}{l}N(x|x^{l})-p(x)|\leq\delta\text{  }
\wedge \text{ } p(x)=0\iff N(x|x^{l})=0\}
\end{equation}
where $N(x|x^{l})$ is the number of occurrences of letter $x$ in word $x^{l}$. \\
For each $l\in\mathbb{N}$, we consider the set of types over alphabet $\cX$, $\cT(\cX,l)$ defined as
\begin{equation*}
\cT(\cX,l):=\{\lambda:T_{\lambda}^{l}\neq\emptyset\},
\end{equation*}
where $T_{\lambda}^{l}=T_{\lambda,0}^{l}$ ($\delta=0$). 
	\subsection{Entanglement transmission codes}\label{etcsection}
	In this section, we prove universal entanglement transmission coding results that are to be combined with suitable classical message transmission codes introduced in the next section. The following lemma is a generalization of random entanglement transmission codes obtained in \cite{boche17} and \cite{boche18}, where a in turn generalization of the decoupling lemma from \cite{klesse} has been obtained. As stated in the following lemma, there are two points to be remarked about these codes. First, the random nature of these codes gives us an encoding state (outcome of the random encoding operation) with a tensor product structure, that is of interest for the present work. Therefore at this stage, we skip the de-randomization step that seemed natural in the original work. Secondly, the integration over unitary groups with respect to the normalized Haar measure done in the random encoding operation therein, is replaced here by an average over the elements of discrete and finite subsets of representations of the unitary group known as unitary designs (see e.g. \cite{evenly}). \\
	The product structure of the encoding state can be used for an instance of channel coding stated later on. This becomes clear when the tensor product structure of the average state is used to accommodate typicality. For $p\in\mathcal{P}(\mathcal{X})$ where $\mathcal{X}$ is some finite alphabet, $\delta>0$ and $x^{l}\in \cX^{l}$, we introduce the following notation. For the tuple $x^{l}:=(x_{1},\dots,x_{l})$ where $x_{i}\in\cX$ for $i=1,\dots,l$, we define 
	\begin{equation*}
	\cG_{x^{l}}:=\cG_{x_{1}}\otimes\dots\otimes\cG_{x_{l}},
	\end{equation*}
	where $\cG_{x_{i}}\subset\cH_{A}$ and clearly, $\cG_{x^{l}}\subset\cH_{A}^{\otimes l}$. Then $\pi_{x^{l}}:=\pi_{\cG_{x^{l}}}$ denotes the maximally mixed state on $\cG_{x^{l}}$ (correspondingly $\pi_{x}$ denotes the maximally mixed state on $\cG_{x}$ for $x\in\cX$), $\Phi_{x^{l}}$ a purification of $\pi_{\cG_{x^{l}}}$ (correspondingly $\Phi_{x}$ denotes a purification of $\pi_{x}$)  and $X_{x^{l}}$ is a unitary design (see Theorem \ref{unitaridesigns}) for $\cU(\cG_{x^{l}})$.
	The following lemma reduces to Theorem 5 of \cite{boche17} when $|\cX|=1$. 
	\begin{lemma}\label{prop}
		Let $\cJ:=\{\cN_{s}\}_{s\in S}\subset\cC(\cH_{A},\cH_{B})$ be any compound quantum channel and alphabet $\cX$ be given. For subspaces $(\cG_{x})_{x\in\cX}$ with $\cG_{x}\subset\cH_{A}, x\in\cX$, probability distribution $p\in\mathcal{P}(\mathcal{X})$ and $\delta>0$, there exists $l_{0}\in\mathbb{N}$, such that for all $l\geq l_{0}$, we find for each $x^{l}\in T_{p,\delta}^{l}$, a subspace $\mathcal{F}_{A,l}\subset\mathcal{G}_{x^{l}}$ and a family $(\mathcal{P}_{i},\mathcal{R}_{i})_{i=1}^{|X_{x^{l}}|}$ of $(l,\dim(\cF_{A,l}))$ entanglement transmission codes with $|X_{x^{l}}|<\infty$ and
		\begin{enumerate}
			\item$\frac{1}{l}\log \dim(\cF_{A,l})\geq\inf_{s\in S}I(A\rangle BX,\omega(\cN_{s},p,\Phi))-\delta
			\ \ $, with $\omega(\cN_{s},p,\Phi) \ $ defined in (\ref{evaluationstate}) for $\Phi:=(\Phi_{x}:x\in\cX)$,
			\item $\forall s\in S\ \ \frac{1}{|X_{x^{l}}|}\sum_{i=1}^{|X_{x^{l}}|}F_{e}(\pi_{\cF_{A,l}},\cR_{i}\circ \cN_{s}^{\otimes l}\circ\cP_{i})\geq 1-\epsilon_{l} \ $ with $\epsilon_{l}\to 0$ exponentially as $l\to\infty$,
			\item$\frac{1}{|X_{x^{l}}|}\sum_{i=1}^{|X_{x^{l}}|}\cP_{i}(\pi_{\mathcal{F}_{A,l}})=\pi_{x^{l}}$.
		\end{enumerate}
	\end{lemma}
	The ingredients to prove this lemma are presented here in form of two lemmas prior to the main proof. The following two lemmas reduce to Lemma 5 and 6 from \cite{boche17}\footnote{see Lemmas \ref{5boche} and \ref{6boche} for the statements.} when $|\cX|=1$. Following these lemmas, we state Theorem \ref{unitaridesigns} based on which we replace the integration with respect to Haar measure, with an average over a subset of the unitary groups called unitary designs. In short, the entanglement transmission codes in \cite{boche17} were derived given a number $l\in\mathbb{N}$ and subspace $\cG^{\otimes l}\subset\cH^{\otimes l}$. Here, we derive codes for a subspace $\cG_{x^{l}}$, with a tensor product structure determined by word $x^{l}$ (see the description above Lemma \ref{prop}).\\
	\begin{lemma}\label{5}
		Let $(\lambda_{x})_{x\in\cA}$ be a probability distribution with $\lambda_{x}>0,\forall x\in\cA$ on an alphabet $\cA$. For $\rho_{x^{l}}:=\bigotimes_{x\in\cA}\rho_{x}^{\otimes N_{x}}, N_{x}:=\lambda_{x}\cdot l\in\mathbb{N}, \rho_{x}\in\cS(\cH) \ \ \forall x\in\cA$ and $\delta\in(0,1/2)$, there exist a real number $\tilde{c}>0$, functions $h:\mathbb{N}\to\mathbb{R}^{+}$, $\phi:(0,1/2)\to\mathbb{R}^{+}$ with $\lim_{l\to\infty}h(l)=0$ and $\lim_{\delta\to 0}\phi(\delta)=0$ and an orthogonal projection $q_{\delta,l}$ satisfying
		\begin{enumerate}
			\item $\tr(\rho_{x^{l}}q_{\delta,l})\geq 1-|\cA|2^{-l(\tilde{c}\delta^{2}-h(l))}$
			\item $q_{\delta,l}\rho_{x^{l}}q_{\delta,l}\leq 2^{-(S(\rho_{x^{l}})-l\phi(\delta))}q_{\delta,l}$.
		\end{enumerate}
		The last inequality implies
		\begin{equation*}
		\parallel q_{\delta,l}\rho_{x^{l}}q_{\delta,l}\parallel_{2}^{2}\leq 2^{-(S(\rho_{x^{l}})-l\phi(\delta))}.
		\end{equation*}
	\end{lemma}
	\begin{proof}
		Let for each $x\in\cA$, $q_{\delta,N_{x}}^{(x)}$ be the frequency typical projection associated with state $\rho_{x}^{\otimes N_{x}}$ in terms of Lemma \ref{5boche}. We show that the projection operator $q_{\delta,l}:=\bigotimes_{x\in\cA}q_{\delta,N_{x}}^{(x)}$ has the properties listed in the statement above. We have 
		\begin{align*}
		\tr(\rho_{x^{l}}q_{\delta,l})&=\tr(\bigotimes_{x\in\cA} \rho_{x}^{\otimes N_{x}}q_{\delta,N_{x}}^{(x)})=\prod_{x\in\cA}\tr(\rho_{x}^{\otimes N_{x}}q_{\delta,N_{x}}^{(x)})\\&\geq \prod_{x\in\cA}(1-2^{-N_{x}(\bar{c}\delta^{2}-h'(N_{x}))})\\&\geq(1-2^{-c_{0}l(\bar{c}\delta^{2}-h'(c_{0}l))})^{|\cA|}\geq 1-|\cA|2^{-c_{0}l(\bar{c}\delta^{2}-h'(c_{0}l))},
		\end{align*}
		where $c_{0}:=\min_{x\in\cA}\lambda_{x}$. Setting $\tilde{c}=c_{0}\bar{c}$ and $h(l)=c_{0}h'(c_{0}l)$, we have the first claim. To see the second claim, we observe that
		\begin{align*}
		q_{\delta,l}\rho_{x^{l}}q_{\delta,l}&\leq \bigotimes_{x\in\cA}q_{\delta,N_{x}}^{(x)}\rho_{x}^{\otimes N_{x}}q_{\delta,N_{x}}^{(x)}\\&\leq \prod_{x\in\cA}2^{-(S(\rho_{x}^{\otimes N_{x}})-N_{x}\phi(\delta))}\bigotimes_{x\in\cA}q_{\delta,N_{x}}^{(x)}\\&=2^{-(S(\rho_{x^{l}})-l(\sum_{x\in\cA}\lambda_{x}\phi(\delta))}q_{\delta,l},
		\end{align*}
		where in the last equality, we have used additivity of von Neumann entropy. We are done.
	\end{proof}
	\begin{lemma}\label{6}
		Let $(\lambda_{x})_{x\in\cA}$ be a probability distribution with $\lambda_{x}>0,\forall x\in\cA$ on an alphabet $\cA$. For each $\mathcal{N}\in \cC(\mathcal{H},\mathcal{K})$, $\delta\in (0,1/2)$, and maximally mixed state $\pi_{x^{l}}:=\bigotimes_{x\in\cA}\pi_{x}^{\otimes N_{x}}, N_{x}=\lambda_{x}\cdot l\in\mathbb{N}$ on some $\mathcal{G}_{x^{l}}\subset\mathcal{H}^{\otimes l}$, there are functions $\gamma:(0,1/2)\to\mathbb{R}^{+}$ and $h:\mathbb{N}\to\mathbb{R}^{+}$ satisfying $\lim_{\delta\to 0}\gamma(\delta)=0$ and $h(l)\searrow 0$ and an operation $\mathcal{N}_{\delta,l}\in \cC^{\downarrow}(\mathcal{H}^{\otimes l},\mathcal{K}^{\otimes l})$, called the reduced operation with respect to $\mathcal{N}$ and $\pi_{x^{l}}$, such that
		\begin{enumerate}
			\item $\tr(\mathcal{N}_{\delta,l}(\pi_{x^{l}}))\geq 1-|\cA|2^{-l(\hat{c}\delta^{2}-h(l))}$, with constant $\hat{c}>0$.
			\item $\mathcal{N}_{\delta,l}$ has a Kraus representation with at most $n_{\delta,l}\leq 2^{S_{e}(\pi_{x^{l}},\mathcal{N}^{\otimes l})+l(\gamma(\delta)+\check{c}h(l))}$ Kraus operators with constant $\check{c}>0$.
			\item For every state $\rho\in\cS(\mathcal{H}^{\otimes l})$ and every two channels $\mathcal{M}\in \cC^{\downarrow}(\mathcal{H}^{\otimes l},\mathcal{H}^{\otimes l})$ and $\mathcal{L}\in \cC^{\downarrow}(\mathcal{K}^{\otimes l},\mathcal{H}^{\otimes l})$, the inequality
			\begin{equation*}
			F_{e}(\rho,\mathcal{L}\circ\mathcal{N}_{\delta,l}\circ\mathcal{M})\leq F_{e}(\rho,\mathcal{L}\circ\mathcal{N}^{\otimes l}\circ\mathcal{M})
			\end{equation*}
			is fulfilled. 
			\item As the set of Kraus operators of $\mathcal{N}_{\delta,l}$ is a subset of the set of Kraus operators of $\cN^{\otimes l}$ for each $l\in\mathbb{N}$, we have
			\begin{equation*}
			\mathcal{N}_{\delta,l}(\sigma)\leq\cN^{\otimes l}(\sigma)\ \ \forall\sigma\in\cS(\cH^{\otimes l}).
			\end{equation*}
		\end{enumerate}
	\end{lemma}
	\begin{proof}
		Let for $x\in\cA$, $\cN_{\delta,N_{x}}^{(x)}$ be the reduced operation for $\pi_{x}^{\otimes N_{x}}$ in terms of Lemma \ref{6boche}. We show that $\cN_{\delta,l}=\bigotimes_{x\in\cA}\cN_{\delta,N_{x}}^{(x)}$ has the properties mentioned above. We have
		\begin{align*}
		\text{tr}(\mathcal{N}_{\delta,l}(\pi_{x^{l}}))&=\prod_{x\in\cA}\text{tr}(\mathcal{N}_{\delta,N_{x}}(\pi_{x}^{\otimes N_{x}}))\geq \prod_{x\in\cA} (1-2^{-N_{x}(c'\delta^{2}-h'(N_{x}))})\\&\geq (1-2^{-c_{0}l(c'\delta^{2}-h'(c_{0}l))})^{|\cA|}\geq 1-|\cA|2^{-c_{0}l(c'\delta^{2}-h'(c_{0}l))},
		\end{align*}
		where $c_{0}:=\min_{x\in\cA}\lambda_{x}$. Setting $h(l)=c_{0}h'(c_{0}l)$ and $\hat{c}=c_{0}c'$ we conclude the first claim. Also the following holds for $n_{\delta,l} $, the number of Kraus operators of $\cN_{\delta,l}$.
		\begin{align*}
		n_{\delta,l}=\bigotimes_{x\in\cA}n_{\delta,N_{x}}&\leq\prod_{x\in\cA}2^{(S_{e}(\pi_{x}^{\otimes N_{x}},\mathcal{N}^{\otimes N_{x}})+N_{x}\gamma(\delta)+N_{x}h'(N_{x}))}\\&\leq 2^{(S_{e}(\pi_{x^{l}},\mathcal{N}^{\otimes l})+l(\sum_{x\in\cA}\lambda_{x}\gamma(\delta)+\frac{\lambda_{x}}{c_{0}}h(l))}\\&=2^{(S_{e}(\pi_{x^{l}},\mathcal{N}^{\otimes l})+l(\gamma(\delta)+\frac{1}{c_{0}}h(l))},
		\end{align*} 
		where in the second line we have used additivity of the entropy exchange $S_{e}$. Finally, the last property comes from multiplicativity of the trace and entanglement fidelity function with respect to tensor products of its arguments. 
	\end{proof}
	We now have generalized statements of Lemmas 5 and 6 from \cite{boche17}. In the statement of Lemma \ref{prop}, we have used unitary designs to mimic the average over the unitary group with respect to Haar measure. The following theorem contains a definition of unitary designs.
	\begin{theorem}\label{unitaridesigns}(See e.g. \cite{evenly})
		Let $\cG$ be a Hilbert space. For unitaries $U\in\mathcal{U}(\mathcal{G})$, there exists a finite set $X\subset\mathcal{U}(\mathcal{G})$ with $|X|\leq\dim(\cG)^{4}$ such that
		\begin{equation}\label{2-design}
		\int_{U\in\mathcal{U}(\mathcal{G})}(U\otimes U) (\cdot)(U\otimes U)^{\dagger}dU=\frac{1}{|X|}\sum_{U\in X}(U\otimes U)(\cdot)(U\otimes U)^{\dagger}
		\end{equation}
		where the integration is with respect to the normalized Haar measure. From this definition it is clear that for $X$ we also have,
		\begin{equation}\label{1-design}
		\int_{U\in\mathcal{U}(\mathcal{G})}U (\cdot)U^{\dagger}dU=\frac{1}{|X|}\sum_{U\in X}U(\cdot)U^{\dagger}.
		\end{equation}
	\end{theorem}
	We refer to the set $X$ as a unitary design. We proceed with the proof.\\
	The expected fidelity function present in \cite{boche17} and \cite{boche18} is achieved by averaging the fidelity function over unitary group with respect to the Haar measure. Here we show that we can replace this by an expected value achieved by taking the average over the unitaries from the relevant unitary design. This brings us to the final statement needed to prove Lemma \ref{prop}, that is an implication of Lemma \ref{decouplemma}. We take the average of both sides of (\ref{lb}) with respect to the unitary design introduced in Theorem \ref{unitaridesigns}, to arrive at the desired expression for the expected fidelity lower-bounded. This result is essentially stated in the proof of Theorem 3.2 \cite{boche18}, to which we refer for more information. In the statement, we will also use the following notation.
	\begin{equation}
	F_{c,e}(\rho,\cN):=\max_{\cR\in\cC(\cH_{B},\cH_{A})}F_{e}(\rho,\cR\circ\cN),
	\end{equation}
	where $\rho\in\cS(\cH_{A})$ and $\cN\in\cC^{\downarrow}(\cH_{A},\cH_{B})$. 
	\begin{lemma}\label{entfid}
		Let $X$ be a unitary design in $\cG$ and $\cF\subset\cG$. With quantities defined as in Lemma \ref{decouplemma}, we have
		\begin{align*}
		\mathbb{E}F_{c,e}(U\pi_{\mathcal{F}}U^{\dagger},\overline{\mathcal{N}}):=&\frac{1}{|X|}\sum_{U\in X}F_{c,e}(U\pi_{\mathcal{F}}U^{\dagger},\overline{\mathcal{N}})\\&\geq\text{tr}(\overline{\mathcal{N}}(\pi_{\mathcal{G}}))
		-2\sum_{j=1}^{|S|}\sqrt{kn_{j}}\parallel\mathcal{N}_{j}(\pi_{\mathcal{G}})\parallel_{2}.
		\end{align*}
	\end{lemma}
	\begin{proof}
		In the first and more straight forward step, we take the average of first term on the right hand side of (\ref{lb}), namely $w_{U}=\text{tr}(\overline{\mathcal{N}}(U\pi_{\mathcal{F}}U^{\dagger}))$;
		\begin{equation}\label{first-term}
		\frac{1}{|X|}\sum_{U\in X}\text{tr}(\overline{\mathcal{N}}(U\pi_{\mathcal{F}}U^{\dagger}))=
		\text{tr}(\overline{\mathcal{N}}(\frac{1}{|X|}\sum_{U\in X}U\pi_{\mathcal{F}}U^{\dagger}))=\overline{\mathcal{N}}(\pi_{\mathcal{G}}).
		\end{equation}			
		What remains is the expected value of $\parallel{D(kU\pi_{\mathcal{F}}U^{\dagger})}\parallel_{1}$. To make the calculation easier we consider averaging of an upper bound on this term in terms of the $2$-norm. From \cite{boche18} we know that
		\begin{equation*}
		\parallel{D(kU\pi_{\mathcal{F}}U^{\dagger})}\parallel_{1}\leq
		\sum_{j,l=1}^{|S|}\frac{1}{|S|}\sqrt{k\min\{n_{j},n_{l}\}\parallel D_{j,l}(kU\pi_{\mathcal{F}}U^{\dagger})\parallel_{2}^{2}}.
		\end{equation*}
		Using the concavity of square root function and Jensen's inequality we have
		\begin{equation*}
		\mathbb{E}(\parallel{D(kU\pi_{\mathcal{F}}U^{\dagger})}\parallel_{1})\leq
		\sum_{j,l=1}^{|S|}\frac{1}{|S|}\sqrt{k\min\{n_{j},n_{l}\} \mathbb{E}(\parallel D_{j,l}(kU\pi_{\mathcal{F}}U^{\dagger})\parallel_{2}^{2})},
		\end{equation*}
		where the expectation is taken over the unitaries belonging to the design. To use Klesse's \cite{klesse} argument as done in proof of Theorem 3.2 of \cite{boche18}, we must invoke the unitary invariance of $\mathbb{E}(\parallel D_{j,l}(kU\pi_{\mathcal{F}}U^{\dagger})\parallel_{2}^{2})$ with respect to all $U\in\mathcal{U}(\mathcal{G})$. To see this unitary invariance, we observe that (see \cite{boche18})
		\begin{equation}\label{summands}
		\parallel D_{j,l}(p)\parallel_{2}^{2}=\frac{1}{k^{2}}\sum_{i=1,r=1}^{n_{j},n_{l}}\text{tr}(p(a^{\dagger}_{j,i}a_{l,r})^{\dagger}pa^{\dagger}_{j,i}a_{l,r})-
		|\text{tr}(pa^{\dagger}_{j,i}a_{l,r})|^{2}.
		\end{equation}
		The unitary invariance of the expectation of the first summand is clear due to linearity of the trace function. For the expectation of the second summand we have
		\begin{align*}
		\frac{1}{|X|}\sum_{U\in X}|\text{tr}(UpU^{\dagger}a^{\dagger}_{j,i}a_{l,r})|^{2}	&=\frac{1}{|X|}\sum_{U\in X}\text{tr}(UpU^{\dagger}a^{\dagger}_{j,i}a_{l,r})\text{tr}(UpU^{\dagger}a^{\dagger}_{l,r}a_{j,i})\\
		&=\frac{1}{|X|}\sum_{U\in X}\text{tr}(UpU^{\dagger}a^{\dagger}_{j,i}a_{l,r}\otimes UpU^{\dagger}a^{\dagger}_{l,r}a_{j,i})\\
		&=\frac{1}{|X|}\sum_{U\in X}\text{tr}(U\otimes U (p \otimes p)(U\otimes U)^{\dagger}(A_{jilr}\otimes A_{jilr}^{\dagger}))\\
		&=\text{tr}(\frac{1}{|X|}\sum_{U\in X} U\otimes U (p \otimes p)(U\otimes U)^{\dagger}(A_{jilr}\otimes A_{jilr}^{\dagger})),
		\end{align*}
		where $A_{jilr}:=a^{\dagger}_{j,i}a_{l,r}$. From (\ref{2-design}), we conclude the invariance of second summand in (\ref{summands}). Therefore we can conclude that $\mathbb{E}(\parallel D_{j,l}(U\pi_{\mathcal{F}}U^{\dagger})\parallel_{2}^{2})$ is indeed invariant with respect to all $U\in\mathcal{U}(\mathcal{G})$. The rest of the proof is exactly the same as the proof of Theorem 3.2 of \cite{boche18}, yet stated here for reader's convenience, as follows. We can use Klesse's argument to conclude
		\begin{equation}\label{secondterm}
		\mathbb{E}(\parallel D_{j,l}(kU\pi_{\mathcal{F}}U^{\dagger})\parallel_{2}^{2})\leq\text{tr}(\mathcal{N}_{j}(\pi_{\mathcal{G}})\mathcal{N}_{l}(\pi_{\mathcal{G}})).
		\end{equation}
		Using (\ref{first-term}), (\ref{lb}) and (\ref{secondterm}) we conclude
		\begin{equation}\label{24}
		\mathbb{E}(F_{c,e}(U\pi_{\mathcal{F}}U^{\dagger},\overline{\mathcal{N}}))\geq\text{tr}(\overline{\mathcal{N}}(\pi_{\mathcal{G}}))
		-\sum_{j,l=1}^{|S|}\frac{1}{|S|}\sqrt{L_{jl}D_{jl}},
		\end{equation}
		where for $j,l\in\{1,...,|S|\}$, we introduce abbreviations
		\begin{equation*}
		L_{j,l}=k\min\{n_{j},n_{l}\}
		\end{equation*}
		and 
		\begin{equation*}
		D_{j,l}=\text{tr}(\mathcal{N}_{j}(\pi_{\mathcal{G}})\mathcal{N}_{l}(\pi_{\mathcal{G}}))
		=\braket{\mathcal{N}_{j}(\pi_{\mathcal{G}}),\mathcal{N}_{l}(\pi_{\mathcal{G}})}_{HS},
		\end{equation*}
		where $\braket{\cdot,\cdot}_{HS}$ denotes the Hilbert Schmidt product. It is obvious that
		\begin{equation*}
		L_{jl}\leq L_{jj} \text{and} L_{lj}\leq L_{ll}.
		\end{equation*}
		Moreover, the Cauchy-Schwartz inequality for the Hilbert-Schmidt inner product justifies the following chain of inequalities.
		\begin{align*}
		D_{jl}=\braket{\mathcal{N}_{j}(\pi_{\mathcal{G}}),\mathcal{N}_{l}(\pi_{\mathcal{G}})}_{HS}
		\leq \parallel\mathcal{N}_{j}(\pi_{\mathcal{G}})\parallel_{2}\parallel\mathcal{N}_{l}(\pi_{\mathcal{G}})\parallel_{2}
		\leq \max\{\parallel\mathcal{N}_{j}(\pi_{\mathcal{G}})\parallel_{2}^{2},\parallel\mathcal{N}_{l}(\pi_{\mathcal{G}})\parallel_{2}^{2}\}
		=\max\{D_{jj},D_{ll}\}.
		\end{align*}
		
		Therefore, an application of Lemma \ref{entriwise} allows us to conclude from (\ref{24}) that
		\begin{equation*}
		\mathbb{E}(F_{c,e}(U\pi_{\mathcal{F}}U^{\dagger},\overline{\mathcal{N}}))\geq\text{tr}(\overline{\mathcal{N}}(\pi_{\mathcal{G}}))
		-2\sum_{j=1}^{|S|}\sqrt{kn_{j}}\parallel\mathcal{N}_{j}(\pi_{\mathcal{G}})\parallel_{2}.
		\end{equation*}
	\end{proof}

Let for $\delta>0$, $\mathcal{N}_{\delta,l,j}$ be the reduced operation associated with $\mathcal{N}_{j},j\in S, |S|<\infty$ as defined by Lemma \ref{6}. Let $q_{\delta,l,j}\in\mathcal{L}(\mathcal{H})$ be the frequency-typical projection of  $\mathcal{N}_{\delta,l,j}(\pi_{x^{l}})$ in terms of Lemma \ref{6}. Define
	\begin{equation}\label{reduceddeltaop}
	\mathcal{N}'_{\delta,l,j}:=\mathcal{Q}_{\delta,l,j}\circ\mathcal{N}_{\delta,l,j}
	\end{equation}
	where $\mathcal{Q}_{\delta,l,j}(\cdot)=q_{\delta,l,j}(\cdot)q_{\delta,l,j}$. Also define
	\begin{equation*}
	\overline{\mathcal{N}}_{\delta,l}:=\frac{1}{|S|}\sum_{j=1}^{|S|}\mathcal{N}'_{\delta,l,j}
	\end{equation*}
	Applying Lemma \ref{entfid} on $\{\mathcal{N}'_{\delta,l,j}\}_{j\in S}$, with expectation taken over unitaries from a unitary design on $\cU(\cG_{x^{l}})$ we obtain
	\begin{equation}\label{RHSLB}
	\mathbb{E}F_{c,e}(U\pi_{\mathcal{F}}U^{\dagger},\overline{\mathcal{N}}_{\delta,l})\geq\text{tr}(\overline{\mathcal{N}}_{\delta,l}(\pi_{x^{l}}))
	-2\sum_{j=1}^{|S|}\sqrt{kn_{\delta,l,j}}\parallel\mathcal{N}'_{\delta,l,j}(\pi_{x^{l}})\parallel_{2}.
	\end{equation}
	We may now follow the steps taken in proof of Theorem 5 from \cite{boche17} to give a lower bound on each of the terms on the right hand side of (\ref{RHSLB}) using Lemmas \ref{5} and \ref{6}, to derive the following result.
	\begin{lemma}\label{premlemma}
		Let $\cJ:=\{\cN_{s}\}_{s\in S}\subset\cC(\cH,\cK)$ be a compound channel, $\delta>0$ and $\lambda\in\cP(\cX)$. For subspaces $(G_{x})_{x\in\cX},\cG_{x}\subset\cH, x\in\cX$, there exists $l_{0}\in\mathbb{N}$ such that for each $l\geq l_{0}$ and $x^{l}\in T_{\lambda}^{l}$, we find a subspace $\cF_{l}\subset\cG_{x^{l}}$ and $(l,\dim(\cF_{l}))$ entanglement transmission codes $(\cP_{i},\cR_{i})_{i=1}^{|X|}$ with $|X|<\infty$ such that,
		
		\begin{enumerate}
			\item $\dim(\cF_{l})\geq 2^{\inf_{s\in S}I_{c}(\pi_{x^{l}},\cN_{s}^{\otimes l})-l\delta}$ and 
			\item $\inf_{s\in S} \frac{1}{|X|}\sum_{i=1}^{|X|}F_{e}(\pi_{\cF_{l}},\cR_{i}\circ\cN_{s}^{\otimes l}\circ\cP_{i})\geq 1-\epsilon_{l} \ $ with $\epsilon_{l}\to 0$ as $l\to\infty$.
		\end{enumerate}
	\end{lemma}
	\begin{proof}
		Let $\cJ_{\tau}$ with index set $S_{\tau}$ be the net associated with $\cJ$ in terms of Lemma \ref{net}. Choose $\delta'\in (0,1/2)$ and $l_{0}\in\mathbb{N}$ satisfying $\gamma(\delta')+\phi(\delta')+\check{c}h(l_{0})\leq\frac{\delta}{2}$ with functions $\gamma,\phi,h$ and constant $\check{c}$ from Lemmas \ref{5} and \ref{6}. Now choose for every $l\geq l_{0}$, a subspace $\cF_{l}\subset\cG_{x^{l}}$ such that
		\begin{equation}\label{dimFl}
		\dim(\cF_{l}):=k_{l}=\lfloor 2^{\min_{s\in S_{\tau}}I_{c}(\pi_{x^{l}},\cN_{s}^{\otimes l})-l\delta}\rfloor.
		\end{equation}
		This is always possible as $S(\pi_{\cG_{x^{l}}})\geq I_{c}(\pi_{x^{l}},\cN_{s}^{\otimes l})$. We have 
		\begin{equation}\label{klbound}
		\min_{s\in S_{\tau}}I_{c}(\pi_{x^{l}},\cN_{s}^{\otimes l})-l\delta-o(l_{0})\leq\log k_{l}\leq\min_{s\in S_{\tau}}I_{c}(\pi_{x^{l}},\cN_{s}^{\otimes l})-l\delta.
		\end{equation}
		We assume for the moment that $x^{l}\in T_{\lambda}^{l}$ is given by concatenation of homogeneous words of size $N_{x}:=N(x|x^{l})$. That is, for $\cA:=\{x\in\cX: N_{x}\neq 0\}\subset\cX$, we have $x^{l}=(x^{N_{x}})_{x\in\cA}$. As such, the hypotheses of Lemma \ref{5} and Lemma \ref{6} apply to to product states indexed by $x^{l}$. This assumption however, does not prohibit generality of the proven results, since each word of type $\lambda$ results from a permutation of the letters of word $x^{l}$. Namely, for any word $\tilde{x}^{l}\in T_{\lambda}^{l}$, there exists a permutation mape $\gamma$ with $\gamma(x^{l})=\tilde{x}^{l}$. Therefore, given codes $(\cP_{i},\cR_{i})_{i\in X}$ for $x^{l}$ with the properties mentioned in the statement of the present lemma, suitable codes for $\tilde{x}^{l}$ will be given by $(\cU_{\gamma}\circ\cP_{i}\circ\cU_{\gamma}^{-1},\cU_{\gamma}^{-1}\circ\cR_{i}\circ\cU_{\gamma})$, with $\cU_{\gamma}$ the CPTP map permuting the tensor factors according to $\gamma$.\\
		We now give lower bounds for the terms on the right hand side of (\ref{RHSLB}).
		\begin{align}\label{tracebound}
		\tr(\overline{\cN}_{\delta',l}(\pi_{x^{l}}))&=\frac{1}{|S_{\tau}|}\sum_{s=1}^{|S_{\tau}|}\tr(\cN'_{\delta',l,s}(\pi_{x^{l}}))\\&=\frac{1}{|S_{\tau}|}\sum_{s=1}^{|S_{\tau}|}\big[\tr(Q_{\delta',l,s}\circ\cN_{s}^{\otimes l}(\pi_{x^{l}}))\nonumber-\tr(Q_{\delta',l,s}\circ[\cN^{\otimes l}-\cN_{\delta',l,s}](\pi_{x^{l}}))\big]\\&\geq 1-|\cX|(2^{-l(\tilde{c}\delta'^{2}-h(l))}-2^{-l(c\delta'^{2}-h(l))}).
		\end{align}
		In the last inequality we have inserted the bounds from Lemmas \ref{5} and \ref{6}, after using $0\leq\tr(Q_{\delta',l,s}\circ[\cN^{\otimes l}-\cN_{\delta',l,s}](\pi_{x^{l}}))\leq\tr([\cN^{\otimes l}-\cN_{\delta',l,s}](\pi_{x^{l}}))$. Also,
		\begin{align}\label{2-normbound}
		\parallel\cN'_{\delta',l,s}(\pi_{x^{l}})\parallel_{2}^{2}&\leq\parallel Q_{\delta',l,s}\circ\cN_{\delta',l,s}(\pi_{x^{l}})\parallel_{2}^{2}\nonumber+\parallel Q_{\delta',l,s}\circ(N_{s}^{\otimes l}-\cN_{\delta',l,s})(\pi_{x^{l}})\parallel_{2}^{2}\\&\leq\parallel Q_{\delta',l,s}\circ\cN_{s}^{\otimes l}(\pi_{x^{l}})\parallel_{2}^{2}\leq 2^{-(S(\pi_{x^{l}})-l\phi(\delta'))}.
		\end{align}
		In the second inequality we have used $\parallel A\parallel_{2}^{2}+\parallel B\parallel_{2}^{2}\leq\parallel A+B\parallel_{2}^{2}$ for non-negative operators $A,B\in\cL(\cK^{\otimes l})$ (see \cite{klesse}), and inserted the lower bound from Lemma \ref{5}. Inserting the bounds from (\ref{tracebound}) and (\ref{2-normbound}) into (\ref{RHSLB}) we obtain
		\begin{align}\label{fcelb}
		\mathbb{E}F_{c,e}(U\pi_{\cF_{l}}U^{\dagger},\overline{\cN}_{\delta',l})&\geq 1-|\cX|\big[2^{-l(c\delta'^{2}-h(l))}-2^{-l(\tilde{c}\delta'^{2}-h(l))}\big]\nonumber\\&-2\sum_{s=1}^{|S_{\tau}|}\sqrt{2^{\log k_{l}-S(\pi_{x^{l}})+l\phi(\delta')+S_{e}(\pi_{x^{l}},\cN_{s}^{\otimes l})+l(\gamma(\delta')+\check{c}h(l))}}\nonumber\\&\geq 1-|\cX|\big[2^{-l(c\delta'^{2}-h(l))}-2^{-l(\tilde{c}\delta'^{2}-h(l))}\big]-2|S_{\tau}|\sqrt{2^{-l(\delta-\phi(\delta')-\gamma(\delta')-\check{c}h(l))}}. 
		\end{align}
		In the second inequality above we have inserted the upper bound for $k_{l}$ from (\ref{klbound}). For $l\geq l_{0}$, (\ref{fcelb}) gives us an exponential decay of error. Therefore we can write
		\begin{align*}
		\mathbb{E}F_{c,e}(U\pi_{\cF_{l}}U^{\dagger},\overline{\cN}_{\delta',l})\geq 1-\epsilon_{1,l}-|S_{\tau}|\epsilon_{2,l}
		\end{align*}
		with $\epsilon_{i,l}\to 0$ with $l\to\infty$ for $i=1,2$. From this we conclude
		\begin{align*}
		\min_{s\in S_{\tau}}\mathbb{E}F_{c,e}(U\pi_{\cF_{l}}U^{\dagger},Q_{\delta',l,s}\circ\cN_{\delta',l,s})\geq 1-|S_{\tau}|\epsilon_{1,l}-|S_{\tau}|^{2}\epsilon_{2,l}.
		\end{align*}
		From the third property under Lemma \ref{6}, the above inequality implies
		\begin{align}\label{Q}
		\min_{ s\in S_{\tau}}\mathbb{E}F_{c,e}(U\pi_{\cF_{l}}U^{\dagger},Q_{\delta',l,s}\circ\cN^{\otimes l}_{s})\geq 1-|S_{\tau}|\epsilon_{1,l}-|S_{\tau}|^{2}\epsilon_{2,l}\geq 1-|S_{\tau}|^{2}\epsilon_{0,l},
		\end{align}
		where $\epsilon_{0,l}:=\max_{i=1,2}\epsilon_{i,l}$. Setting shorthand notation $\beta_{s,U}:=1-F_{c,e}(U\pi_{\cF_{l}}U^{\dagger},Q_{\delta',l,s}\circ\cN^{\otimes l}_{s})$, we obtain from  Lemma \ref{projectionlemma}, $\ F_{c,e}(U\pi_{\cF_{l}}U^{\dagger},
		\cN^{\otimes l}_{s})\geq 1-3\beta_{s,U}$. Hence from (\ref{Q}) we conclude
		\begin{align}\label{Q1}
		\min_{s\in S_{\tau}}\mathbb{E}F_{c,e}(U\pi_{\cF_{l}}U^{\dagger},\cN^{\otimes l}_{s})\geq  1-3|S_{\tau}|^{2}\epsilon_{0,l}.
		\end{align}
		By Lemma \ref{net}, we have 
		\begin{align}\label{Q2}
		\min_{ s\in S}\mathbb{E}F_{c,e}(U\pi_{\cF_{l}}U^{\dagger},\cN^{\otimes l}_{s})\geq  1-3|S_{\tau}|^{2}\epsilon_{0,l}-2l\tau.
		\end{align}
		
		Given that we find $|S_{\tau}|\leq(\frac{6}{\tau})^{2(d\cdot d')^{2}}$, choosing $\tau=\epsilon_{0,n}^{\frac{1}{8(d\cdot d')^{2}}}$, we have the desired exponential decay of error. Also, as $\cJ_{\tau}\subset\cJ$, we obtain the desirable lower bound on the rate. 
	\end{proof}
	\begin{proof}[Proof of Lemma \ref{prop}]
		Assume $x^{l}\in T_{\lambda}^{l}$. Note that $\forall\lambda\in\mathcal{T}(\mathcal{X},l)$, either $T_{\lambda}^{l}\subset T_{p,\delta}^{l}$ or $T_{\lambda}^{l}\bigcap T_{p,\delta}^{l}=\emptyset$. Since by assumption of the lemma $x^{l}\in T_{p,\delta}^{l}$, we conclude $T_{\lambda}^{l}\subset T_{p,\delta}^{l}$. For each $\tilde{\delta}>0$, we have from Lemma \ref{premlemma} applied on the compound channel $\cJ\subset\cC(\cH_{A},\cH_{B})$, a subspace $\cF_{A,l}\subset\cH_{A}^{\otimes l}$ with
		\begin{align*}
		\frac{1}{l}\log\dim(\cF_{A,l})\geq\frac{1}{l}\inf_{s\in S}I_{c}(\pi_{x^{l}},\cN_{s}^{l})-\tilde{\delta}=&\inf_{s\in S}\sum_{x\in\cA}\frac{1}{l}I_{c}(\pi_{x}^{\otimes N_{x}},\cN_{s}^{\otimes N_{x}})-\tilde{\delta}\\&=\inf_{s\in S}\sum_{x\in\cA}\lambda(x)I_{c}(\pi_{x},\cN_{s})-\tilde{\delta}\\&\geq\inf_{s\in S}\sum_{x\in\cA}p(x)(I_{c}(\pi_{x},\mathcal{N}_{s})-\tilde{\delta})-|\lambda(x)-p(x)|
		\cdot(I_{c}(\pi_{x},\mathcal{N}_{s})-\tilde{\delta})
		\\&\geq\inf_{s\in S}I(A\rangle BX,\omega(\cN_{s},p,\Phi))-\tilde{\delta}-|\cX|\overline{c}\tilde{\delta}.
		\end{align*}
		with$\overline{c}:=2\log\dim(\cH_{A}\otimes\cH_{B})$ and $\omega_{s}$ defined by (\ref{evaluationstate}). With this rate we obtain exponential decay of error as explained above. Choosing $\tilde{\delta}$ such that  $\delta>\tilde{\delta}+|\cX|\overline{c}\tilde{\delta}$, we obtain the desired lower bound on the rate. The last property listed under Lemma \ref{prop} is clear by averaging property of the Haar measure, reproduced here by the unitary design in $\cU(\cG_{x^{l}})$.
	\end{proof} 
	\subsection{Classical message transmission codes}\label{ctresult}
	The desired statement of universal codes for c-q channels can be extracted from \cite{mosonyi}. Therein, the authors have introduced universal random codes for transmission of classical messages over c-q channels, using properties of Renyi entropies. Based on the same codes, we have derived the following lemma to allow for a faster decay of error while considering only " typical " inputs.  
	\begin{lemma}\label{corollary22}
		Let $\cJ:=\{\cN_{s}\}_{s\in S}\subset\cC(\cH_{A},\cH_{B})$ and $V:\cX\to\cS(\cH_{A})$ be a c-q channel. For each $\eta>0$ and $p\in\mathcal{P}(\mathcal{X})$, there exists a number $n_{0}$, such that for $n\geq n_{0}$, there exists a classical encoding map $u:m\to u_{m}\in\mathcal{X}^{n}$ and decoding POVM $(\Lambda)_{m\in[M_{n}]}$ such that 
		\begin{enumerate}
			\item	$\forall m\in[M_{n}]:u_{m}\in T_{p,\eta}^{n}$,
			\item$\inf_{s\in S}\min_{m\in M_{n}}\tr((\cN_{s}\circ V)^{\otimes n}(u_{m})\Lambda_{m})\geq 1-\epsilon_{n}$, with $\epsilon_{n}\to 0$ exponentially as $n\to\infty$,
			\item $\frac{1}{n}\log M_{n}\geq \inf_{s\in S}I(X;B,\omega(\cN_{s},p,\Psi))-c\eta$
		\end{enumerate}
		with $\omega(\cN_{s},p,\Psi)$ defined by (\ref{evaluationstate}) for $\Psi:=(\Psi_{x}:x\in\cX,\tr_{\cH}(\Psi_{x})=V(x))$ and constant $c>0$. 
	\end{lemma}
	\section{Proofs for the compound channel}\label{proofsforcompound}
	In this section we proceed with the proof of Theorem \ref{mainresult} in two parts. In the following section (converse part), it is also demonstrated that CSI at the decoder does not improve channel's classically enhanced entanglement generation capacity. In the more involved direct part of the proof, we introduce classically enhanced entanglement transmission codes by marrying classical message transmission codes from \cite{mosonyi} and a generalization of entanglement transmission codes from \cite{boche17} and \cite{boche18} as stated in Section \ref{preliminaryresults}. 
	
\subsection{Proof of the converse }\label{conversesection}
	In this section we prove the following lemma. 
	\begin{lemma}\label{lemmaconv}
		Let $\mathcal{J}:=\{\cN_{s}\}_{s\in S}\subset \cC(\mathcal{H}_{A},\mathcal{H}_{B})$ be any compound quantum channel. It holds
		\begin{equation}\label{converse}
		\overline{C}_{CEG}(\mathcal{J})\subset
		\cl\bigg(\bigcup_{l=1}^{\infty}\frac{1}{l}\bigcup_{p,\Psi}\bigcap_{s\in S}\hat{C}(\mathcal{N}_{s}^{\otimes l},p,\Psi)\bigg).
		\end{equation}
	\end{lemma}
		To prove this result, we shall make use of the following lemma (see \cite{idev}).
	\begin{lemma}\label{lemma3}
		For two states $\rho^{AB}$ and $\sigma^{AB}$ on some Hilbert space $\cK_{A}\otimes\cK_{B}$ of dimension $r$ and fidelity $f:=F(\rho^{AB},\sigma^{AB})$, we have
		\begin{equation*}
		|I(A\rangle B,\rho)-I(A\rangle B,\sigma)|\leq\frac{2}{e}+4\log r \sqrt{1-f}.
		\end{equation*}
	\end{lemma}
	\begin{proof}[Proof of Lemma \ref{lemmaconv}]
		We prove a more general claim than stated in Lemma \ref{lemmaconv}, allowing the decoder to choose the processing according to the channel state (i.e. the decoder has access to CSI). Let for each $n\in\mathbb{N}$, $\mathcal{C}_{CEG,s}:=(\Psi_{m},\mathcal{R}_{m,s})_{m\in M_{1,n}}$ be an $(n,M_{1,n},M_{2,n})$ CEG code with informed decoder\footnote{As clear from the notation, these codes are CEG codes for compound channel $\cJ$, when the decoder has access to CSI.}, such that
		\begin{equation}\label{simerr}
		\inf_{s\in S}\overline{P}(\mathcal{C}_{CEG,s},\mathcal{N}_{s}^{\otimes n})\geq 1-\epsilon,
		\end{equation}
		with $\epsilon<1$ holds. Fix $n\in\mathbb{N}$ and let $p_{*}\in\cP(\cX^{n})$ be the equidistribution on the message set. Consider the pair $(M_{s},M_{s}')$ of random variables with joint distribution:
		\begin{equation*}
		\text{Pr}(M_{s}=m,M_{s}'=m')=p_{*}(m)\text{tr}(\mathcal{R}_{m',s}\circ\mathcal{N}_{s}^{\otimes n}(V(m)))
		\end{equation*}
		for $(m_{s},m'_{s}\in[M_{1,n}])$ and $s\in S$ with $V(m):=\text{tr}_{\cF_{A,n}}\Psi_{m}$ for some c-q channel $V:\cX^{n}\to\cS(\cH_{A}^{\otimes n})$. Note that with these definitions, we have
		\begin{equation}\label{averageclasserror}
		\mathbb{P}(M_{s}\neq M'_{s})\leq 1-\overline{P}(\mathcal{C}_{CEG,s},\mathcal{N}^{\otimes n}_{s})\leq\epsilon,
		\end{equation}
		for $s\in S$. Fix $s$ for the moment. Define the state
		\begin{equation*}
		\sigma'_{s}:=\sum_{m\in[M_{1,n}]}p_{*}(m)\ket{m}\bra{m}^{X}\otimes(\text{id}_{\mathcal{H}_{A}^{\otimes n}}\otimes\mathcal{R}_{s,m}\circ\mathcal{N}_{s}^{\otimes n})(\Psi_{m})
		\end{equation*}
		and the shorthand notation
		\begin{equation*}
		\sigma_{s}:=\omega(\cN_{s}^{\otimes n},p_{*},\Psi)=\sum_{m\in[M_{1,n}]}p_{*}(m)\ket{m}\bra{m}^{X}\otimes(\text{id}_{\mathcal{H}_{A}^{\otimes n}}\otimes\mathcal{N}_{s}^{\otimes n})(\Psi_{m}).
		\end{equation*}
		 We have
		\begin{align}\label{22}
		\log M_{1,n}=H(p_{*})=I(M_{s};M'_{s})+H(M'_{s}|M_{s})
		&\leq I(M_{s};M'_{s})+\epsilon\log M_{1,n}+1\nonumber\\&\leq I(X;B,\sigma_{s})+
		\epsilon\log M_{1,n}+1\nonumber\\&
		\leq I(X;B,\sigma_{s})+n\epsilon\log|\cX|+1,
		\end{align}
		where $I(Y;Y')$ is the mutual information of random variables $Y,Y'$. The first inequality comes from (\ref{averageclasserror}) and the second is by Holevo bound (see \cite{wilde13}). For $s\in S$, 
		We have
		\begin{align}\label{25}
		\epsilon&\geq 1-\overline{P}(\mathcal{C}_{CEG,s},\mathcal{N}_{s}^{\otimes n})=
		1-F(\Phi,\sigma_{s}'^{AB}),
		\end{align}
		where $\sigma_{s}'^{AB}:=\tr_{X}(\sigma_{s}')$. We have
		\begin{align}\label{dataprocessing}
		I(A\rangle BX:\sigma_{s})&\geq I(A\rangle BX,\sigma_{s}')\nonumber
		\\
		&\geq I(A\rangle B,\sigma_{s}'^{AB})\nonumber\\
		&\geq I(A\rangle B,\Phi)-\frac{2}{e}-8n\log\dim\mathcal{H}\sqrt{\epsilon}=\log M_{2,n}-\frac{2}{e}-8n\log\dim\mathcal{H}\sqrt{\epsilon}. 
		\end{align}
		In (\ref{dataprocessing}), the first inequality comes from the quantum data processing inequality, the second comes from the fact that conditioning does not decrease coherent information. The third inequality comes from Lemma \ref{lemma3} together with (\ref{25}) and finally, in the last line we have used $I(A\rangle B,\Phi)=\log M_{2,n}$.\\
		Choosing $n$ such that  $\delta\geq\frac{2}{ne}+8\log\dim(\cH)\sqrt{\epsilon}$ , from (\ref{dataprocessing}) and (\ref{22}) we obtain
		\begin{equation*}
		(\frac{1}{n}\log M_{1,n}-\delta,\frac{1}{n}\log M_{2,n}-\delta)\in\frac{1}{n}\hat{C}(\mathcal{N}_{s}^{\otimes n}, p_{*},\Psi).
		\end{equation*}
	
Since $s\in S$ was arbitrary, we have shown
		\begin{align}\label{converseM}
		(\frac{1}{n}\log M_{1,n}-\delta,\frac{1}{n}\log M_{2,n}-\delta)\in\cl[\bigcup_{n=1}^{\infty}\bigcup_{p,\Psi}\bigcap_{s\in S}\frac{1}{n}\hat{C}(\mathcal{N}_{s}^{\otimes n}, p,\Psi)].
		\end{align} 
	\end{proof}
	\subsection{Proof of the direct part}\label{directsection}
	In this section we prove the following lemma.
		\begin{lemma}\label{dp}
			Let $\mathcal{J}:=\{\cN_{s}\}_{s\in S}\subset \cC(\mathcal{H}_{A},\mathcal{H}_{B})$ be any compound quantum channel. It holds
			\begin{equation}\label{direct}
			\cl\bigg(\bigcup_{l=1}^{\infty}\frac{1}{l}\bigcup_{p,\Psi}\bigcap_{s\in S}\hat{C}(\mathcal{N}_{s}^{\otimes l},p,\Psi)\bigg)\subset C_{CET}(\mathcal{J}).
			\end{equation}
			\end{lemma}
	In the first step towards proving the above statement, we restrict the encoder to maximally mixed state inputs. The final result will then be a generalization by way of which we lift this restriction. We state the first instance of the classically enhanced codes, satisfying classical and quantum error criteria in the following lemma. 
	\begin{lemma}\label{piinput}
		Let $\cJ:=\{\cN_{s}\}_{s\in S}\subset\cC(\cH_{A},\cH_{B})$ be any quantum compound channel. For finite alphabet $\cX$, subspaces $(\cG_{x})_{x\in\cX},\cG_{x}\subset\cH_{A}$ $\forall x\in\cX$,  $p\in\cP(\cX), V_{\pi}:\cX\to\cS(\cH_{A})$ with $ V_{\pi}(x)=\pi_{x}, x\in\cX$, each $\delta>0$ and large enough values of $n$, there exists an $(n,M_{1,n}, M_{2,n})$ CET code with $M_{2,n}=\dim(\cF_{A,n})$ such that
		\begin{enumerate}
			\item $\frac{1}{n}\log M_{2,n}\geq\inf_{s\in S}I(A\rangle BX,\omega(\cN_{s},p,\Phi))-\delta$,
			\item $\frac{1}{n}\log M_{1,n}\geq\inf_{s\in S}I(X;B,\omega(\cN_{s},p,\Phi))-c\delta$ with some constant $c>0$ and $\omega(\cN_{s},p,\Phi)$ defined by (\ref{evaluationstate}) for $\Phi:=(\Phi_{x}:x\in\cX)$ defined as in Section \ref{etcsection},
			\item $\inf_{s\in S}\min_{m\in[M_{1,n}]}\ P(\cC_{CET},\cN_{s}^{\otimes n},m)\geq 1-\epsilon_{n}$,
			with $\epsilon_{n}\to 0$ exponentially as $n\to\infty$. 
		\end{enumerate}
	\end{lemma}
	\begin{proof}
		Let $\cJ_{\tau}\subset\cJ$ be as defined in Appendix \ref{mosonyicodes}, Lemma \ref{net} with index set $S_{\tau}$. According to Lemma \ref{corollary22}, for $\delta>0$ and large enough values of $n\in\mathbb{N}$, we find pairs $(u_{m},\Lambda_{m})_{m\in [M_{1,n}]}$ with $\frac{1}{n}\log M_{1,n}\geq\min_{s\in S_{\tau}}I(X;B,\omega(\cN_{s},p,\Phi))-c\delta$, such that for channel $V_{\pi}$ we have 
		\begin{equation}\label{clascode}
		\min_{s\in S_{\tau}}\min_{m\in[M_{1,n}]}\text{tr}(\Lambda_{m}(\cN_{s}\circ V_{\pi})^{\otimes n}(u_{m}))
		=	\min_{s\in S_{\tau}}\min_{m\in[M_{1,n}]}\text{tr}(\Lambda_{m}\circ\cN_{s}^{\otimes n}(\pi_{u_{m}}))\geq 1-\epsilon_{1,n},
		\end{equation}
		for $u_{m}\in T_{p,\delta}^{n}$ and $\epsilon_{1,n}$ going to zero exponentially. Given $u_{m}\in T_{p,\delta}^{n}$ for each $m$, according to Lemma \ref{prop}, there exists a family of entanglement transmission codes $(\cP_{i}^{(m)},\tilde{\cR_{i}}^{(m)})_{i=1}^{|X_{u_{m}}|}$ with rate $\frac{1}{n}\log M_{2,n}\geq\\\min_{s\in S_{\tau}}I(A\rangle BX,\omega(\cN_{s},p,\Phi))-\delta$, such that  $\pi_{u_{m}}$ is exactly the output of the average of encoding operations (third statement of the lemma) and 
		\begin{equation}\label{quantumerror}
		\min_{s\in S_{\tau}}\min_{m\in[M_{1,n}]}\frac{1}{|X_{u_{m}}|}\sum_{i=1}^ {|X_{u_{m}}|}F_{e}(\pi_{F_{A,n}},\tilde{\mathcal{R}}_{i}^{(m)}\circ\mathcal{N}_{s}^{\otimes n}\circ\mathcal{P}^{(m)}_{i})\geq 1-\epsilon_{2,n}
		\end{equation}
		with $\epsilon_{2,n}\to 0$ exponentially. Thus (\ref{clascode}) yields
		\begin{equation}\label{14}
			\min_{s\in S_{\tau}}\min_{m\in[M_{1,n}]}\frac{1}{|X_{u_{m}}|}\sum_{i=1}^{|X_{u_{m}}|}\text{tr}(\Lambda_{m}\cN_{s}^{\otimes n}(\cP_{i}^{(m)}(\pi_{\cF_{A,n}}))\geq 1-\epsilon_{1,n}.
		\end{equation}

		Following \cite{devetak}, the encoding and decoding maps are given by $(\cP_{i}^{(m)},\cR_{i}^{(m)})_{i=1}^{|X_{u_{m}}|}$ with
		\begin{equation*}
 \cR_{i}^{(m)}(\rho)=\tilde{\cR_{i}}^{(m)}(\sqrt{\Lambda_{m}}\rho\sqrt{\Lambda_{m}}).
		\end{equation*}
		It can be observed that for each $i$ we have $\sum_{m\in[M_{1,n}]}\cR_{i}^{(m)}\in\cC(\cH_{B}^{\otimes n},\cF_{B,n})$. 
		
		From (\ref{14}) we obtain
			\begin{equation}\label{expectcerror}
				\min_{s\in S_{\tau}}\min_{m\in[M_{1,n}]} \frac{1}{|X_{u_{m}}|}\sum_{i=1}^{|X_{u_{m}}|}\text{tr}(\cR_{i}^{(m)}\circ\cN^{\otimes n}_{s}(\cP_{i}^{(m)}(\pi_{\cF_{A,n}})))\geq 1-\epsilon_{1,n}.
			\end{equation}
		We define the following state
		\begin{equation*}
		\chi_{i,s}^{(m)}:= [\text{id}\otimes(\mathcal{N}_{s}^{\otimes n}\circ\mathcal{P}_{i}^{(m)})](\Phi_{\cF_{A,n}}),
		\end{equation*}
		where $\Phi_{\cF_{A,n}}$ is a maximally entangled state given by  purification of $\pi_{\cF_{A,n}}$. From (\ref{14}) we obtain
		\begin{equation}\label{gammaerror}
			\min_{s\in S_{\tau}}\min_{m\in[M_{1,n}]}\frac{1}{|X_{u_{m}}|}\sum_{i=1}^{|X_{u_{m}}|}\text{tr}\chi_{i,s}^{(m)}(\text{id}\otimes\Lambda_{m})\geq 1-\epsilon_{1,n}.
		\end{equation}
		Set $\gamma_{i,s,m}:=\tr\chi_{i,s}^{m}(\id\otimes\Lambda_{m})$. It is clear that if $\gamma_{i,s,m}=0$, we have (\ref{termwise}). To prove this equation for the case where $\gamma_{i,s,m}>0$, we observe that by the gentle measurement lemma (Lemma \ref{gmeasurement}), we have for all $i,m,s$
		\begin{equation*}
		\parallel\frac{(\id\otimes\sqrt{\Lambda_{m}})(\chi_{i,s}^{m})(\id\otimes\sqrt{\Lambda_{m}})}{\gamma_{i,s,m}}-\chi_{i,s}^{(m)}\parallel_{1}\leq 2\sqrt{1-\gamma_{i,s,m}}
		\end{equation*} 
	 and hence by monotonicity of trace distance under CPTP maps we obtain
		\begin{equation}\label{closeness}
		\parallel\frac{1}{\gamma_{i,s,m}}(\text{id}\otimes\mathcal{R}^{(m)}_{i})(\chi_{i}^{(m)})-(\text{id}\otimes\tilde{\mathcal{R}}^{(m)}_{i})(\chi_{i}^{(m)})\parallel_{1}\leq 
		2\sqrt{1-\gamma_{i,s,m}}.
		\end{equation}
		Applying Lemma \ref{appendixlemma} and averaging with respect to index $i$, the above inequality yields
		\begin{align}\label{termwise}
\frac{1}{|X_{u_{m}}|}\sum_{i=1}^{|X_{u_{m}}|}&	F_{e}(\pi_{\cF_{A,n}},\mathcal{R}_{i}^{(m)}\circ\mathcal{N}_{s}^{\otimes n}\circ\mathcal{P}_{i}^{(m)})\geq\nonumber\\&\frac{1}{|X_{u_{m}}|}\sum_{i=1}^ {|X_{u_{m}}|} (F_{e}(\pi_{\cF_{A,n}},\tilde{\mathcal{R}}_{i}^{(m)}\circ\mathcal{N}_{s}^{\otimes n}\circ\mathcal{P}^{(m)}_{i})-2\sqrt{1-\gamma_{i,s,m}})\gamma_{i,s,m}.
		\end{align}
To give a suitable lower bound for (\ref{termwise}), we use Lemma \ref{2epsilonlemma}. We observe that,
		\begin{align}\label{allmis}
		\frac{1}{|X_{u_{m}}|}&\sum_{i=1}^{|X_{u_{m}}|} F_{e}(\pi_{\cF_{A,n}},\tilde{\mathcal{R}}_{i}^{(m)}\circ\mathcal{N}_{s}^{\otimes n}\circ\mathcal{P}^{(m)}_{i})-2\sqrt{1-\gamma_{i,s,m}}\geq\nonumber\\&\frac{1}{|X_{u_{m}}|}\sum_{i=1}^ {|X_{u_{m}}|} F_{e}(\pi_{\cF_{A,n}},\tilde{\mathcal{R}}_{i}^{(m)}\circ\mathcal{N}_{s}^{\otimes n}\circ\mathcal{P}^{(m)}_{i})-2\sqrt{1-\frac{1}{|X_{u_{m}}|}\sum_{i=1}^{|X_{u_{m}}|} \gamma_{i,s,m}}\nonumber\\&\geq 1-\epsilon_{2,n}-2\sqrt{\epsilon_{1,n}},
		\end{align}
		where in the first inequality we have used concavity of the square function along with Jensen's inequality, and in the second one we have used the bounds from (\ref{quantumerror}) and (\ref{gammaerror}). Setting $\epsilon_{3,n}:=\max\{\epsilon_{2,n}-2\sqrt{\epsilon_{1,n}},\epsilon_{1,n}\}$, by Lemma \ref{2epsilonlemma}, (\ref{allmis}), (\ref{gammaerror}) and (\ref{termwise}) imply
		\begin{equation}\label{expectation of q.error}
		\frac{1}{|X_{u_{m}}|}\sum_{i=1} ^{|X_{u_{m}}|}F_{e}(\pi_{\cF_{A,n}},\mathcal{R}_{i}^{(m)}\circ\mathcal{N}_{s}^{\otimes n}\circ\mathcal{P}^{(m)}_{i})\geq 1-2\epsilon_{3,n}.
		\end{equation}
		This means that for each $m$ there exists a value $i(m)$ such that: 
		\begin{equation*}
		\frac{1}{|S_{\tau}|}\sum_{s\in S_{\tau}}F_{e}(\pi_{\cF_{A,n}},\mathcal{R}_{i(m)}^{(m)}\circ\mathcal{N}_{s}^{\otimes n}\circ\mathcal{P}^{(m)}_{i(m)})\geq 1-2\epsilon_{3,n}.
		\end{equation*}
	 Therefore setting $\cR:=\sum_{m\in[M_{1,n}]}\ket{m}\bra{m}\otimes\cR_{i(m)}^{(m)}$ and $\cP_{m}:=\cP_{i(m)}^{(m)}$ for all $m\in[M_{1,n}]$ for all $s\in S_{\tau}$ and $m\in[M_{1,n}]$, we have for $\cC_{CET}:=(\cP_{m},\cR_{m})_{m\in[M_{1,n}]}$ with
	
			\begin{align}\label{quantum error}
			P(\cC_{CET},\cN_{s}^{\otimes n},m)&=F(\ket{m}\bra{m}\otimes\Phi^{AB},\id_{\cF_{A,n}}\otimes\cR\circ\cN_{s}^{\otimes n}\circ\cP_{m})\nonumber\\&=F_{e}(\pi_{\cF_{A,n}},\mathcal{R}_{i(m)}^{(m)}\circ\mathcal{N}_{s}^{\otimes n}\circ\mathcal{P}^{(m)}_{i(m)})\geq 1-2|S_{\tau}|\epsilon_{3,n}.
			\end{align}
By the third property of $\cJ_{\tau}$ stated under Lemma \ref{net}, we have for all $s\in S$ and $m\in[M_{1,n}]$
		\begin{equation*}
		P(\cC_{CET},\cN_{s}^{\otimes n},m)\geq 1-2|S_{\tau}|\epsilon_{3,n}-2n\tau,
		\end{equation*}
			Given that we find $|S_{\tau}|\leq(\frac{6}{\tau})^{2(d\cdot d')^{2}}$, choosing $\tau=\epsilon_{3,n}^{\frac{1}{4(d\cdot d')^{2}}}$, we have the desired exponential decay of error. Also we obtain the desirable rates as $\cJ_{\tau}\subset\cJ$. 
	\end{proof}
	We now run an instance of concatenation upon codes from Lemma \ref{piinput}, to achieve suitable codes without the restriction imposed by $V_{\pi}$. The method used here for lifting this restriction is rather elementary\footnote{Compare  with BSST type lemmas used for instance in \cite{boche17}} given that the input state can be decomposed as a convex combination of maximally mixed states.  
	\begin{lemma}\label{generalinput0}
		For compound channel $\cJ:=\{\cN_{s}\}_{s\in S}\subset\cC(\cH_{A},\cH_{B})$,  $p\in\cP(\cX)$, $V:\cX\to\cS(\cH_{A})$ and large enough values of $n$, there exists a CET codes, $\cC_{CET}:=(\cP^{(m)},\cR^{(m)})_{m\in[M_{1,n}]}$ such that
		\begin{enumerate}
			\item 
			$
			\liminf_{n\to\infty}\frac{1}{n}\log M_{2,n}\geq\inf_{s\in S}I(A\rangle BX,\omega_{s}(\cN_{s},p,\Psi)),
			$
			\item
			$
			\liminf_{n\to\infty}\frac{1}{n}\log M_{1,n}\geq\inf_{s\in S}I(A;B,\omega_{s}(\cN_{s},p,\Psi))
			$ hold with $\omega_{s}(\cN_{s},p,\Psi)$ defined by (\ref{evaluationstate}),
			\item
			$\inf_{ s\in S}\min_{m\in[M_{1,n}]}\ P(\cC_{CET},\cN^{\otimes n},m)\geq 1-\epsilon_{n}$,
		\end{enumerate}
		with $\epsilon_{n}\to 0$ exponentially as $n\to\infty$.
	\end{lemma}
	
	\begin{proof}
		For $x\in\cX$, let $V(x)$ have the spectral decomposition
		\begin{equation*}
		V(x)=\sum_{y\in\cY}q_{x}(y)\ket{\phi_{x}^{y}}\bra{\phi_{x}^{y}},
		\end{equation*}
		with $\cY$ an alphabet with $|\cY|=\dim(\cH_{A})$, $\{\ket{\phi_{x}^{y}}\}_{y\in\cY}$ an ONB and $q_{x}\in\cP(\cY)$ for each $x\in\cX$. It can be seen that for $l\in\mathbb{N}$ and $x^{l}\in\cX^{l}$ we have
		\begin{equation}\label{spectraldecomp}
		V^{\otimes l}(x^{l})=\sum_{y^{l}\in \cY^{l}}q_{x^{l}}(y^{l})\ket{\phi_{x^{l}}^{y^{l}}}\bra{\phi_{x^{l}}^{y^{l}}}.
		\end{equation}
		For each $x^{l}\in\cX^{l}$ and $\lambda\in\cT(\cX\times\cY,l)$, define the following sets
		\begin{equation}\label{jointtypicalsets}
		\cA_{\lambda}(x^{l}):=\{y^{l}:(x^{l},y^{l})\in T_{\lambda}^{l}\}.
		\end{equation}
		Given the properties of typical sets, it can be observed that $\cA_{\lambda}(x^{l})\bigcap \cA_{\lambda'}(x^{l})=\emptyset$ for all pairs $(\lambda,\lambda')$ with $\lambda\neq\lambda'$. Also, $\bigcup_{\lambda\in\cT(\cX\times\cY,l)}\cA_{\lambda}(x^{l})=\cY^{l}$. Given these properties, from (\ref{spectraldecomp}) we obtain
		\begin{align}
V^{\otimes l}(x^{l})=\sum_{\lambda\in\cT(\cX\times\cY,l)}q_{x^{l}}(\lambda)\sum_{y^{l}\in\cA_{\lambda}(x^{l})}\ket{\phi_{x^{l}}^{y^{l}}}\bra{\phi_{x^{l}}^{y^{l}}}=\sum_{\lambda\in\cT(\cX\times\cY,l)}q_{x^{l}}(\lambda)\pi_{x^{l}}^{\lambda},
		\end{align}
with $\pi_{x^{l}}^{\lambda}:=\frac{1}{|\cA_{\lambda}(x^{l})|}\sum_{y^{l}\in\cA_{\lambda}(x^{l})}\ket{\phi_{x^{l}}^{y^{l}}}\bra{\phi_{x^{l}}^{y^{l}}}$ and $q_{x^{l}}(\lambda)=q_{x^{l}}(y^{l})|\cA_{\lambda}(x^{l})|$ for any $y^{l}\in\cA_{\lambda}(x^{l})$. The above decomposition therefore comes from the fact that for all $y^{l}\in\cA_{\lambda}(x^{l})$, $q_{x^{l}}(y^{l})$ is constant. Define probability distribution $r\in\cP(\cX^{l},\cT(\cX\times\cY,l))$ with $r(x^{l},\lambda)=p^{l}(x^{l})q_{x^{l}}(\lambda)$. Also define the state
		\begin{equation*}
		\sigma_{s}:=\sum_{(x^{l},\lambda)\in\cX^{l}\times\cT(\cY,l)}r(x^{l},\lambda)\ket{e_{x^{l}}}\bra{e_{x^{l}}}^{X}\otimes\ket{e_{\lambda}}\bra{e_{\lambda}}^{T}\otimes\mathbbm{1}_{\cH_{A}^{\otimes l}}\otimes\cN_{s}^{\otimes l}(\Phi_{x^{l}}^{\lambda}),
		\end{equation*}
		where $\Phi_{x^{l}}^{\lambda}$ is a purification of $\pi_{x^{l}}^{\lambda}$, a maximally entangled state on subspace $\cG_{x^{l}}^{\lambda}\subset\cH_{A}^{\otimes l}$. According to Lemma \ref{piinput}, for $V_{\pi}:V_{\pi}(x^{l},\lambda)=\pi_{x^{l}}^{\lambda}$, large enough values of $a\in\mathbb{N}$ and $\delta>0$, we find a subspace $\cF_{A,a\cdot l}\subset\cH_{A}^{\otimes a\cdot l}$ with $\dim(\cF_{A,a\cdot l})=M_{2,a\cdot l}$ with
		\begin{align}\label{quantumrategen}
		\frac{1}{a}\log M_{2,a\cdot l}&\geq\inf_{s\in S}I(A\rangle B TX,\sigma_{s})-\delta\\&\geq\inf_{s\in S}I(A T\rangle B X,\sigma_{s})-\delta\\&\geq\inf_{s\in S}I(A\rangle B X,(\sigma_{s}^{XAB})^{\otimes l})-S(T)_{\sigma}-\delta\nonumber\\&\geq\inf_{s\in S}I(A\rangle BX,(\sigma_{s}^{XAB})^{\otimes l})-\dim(\cH_{A}\otimes\cH_{B})\log(l+1)-\delta.
		\end{align}
		The first inequality comes from an application of Lemma \ref{piinput}, second and third from well-known inequalities (see e.g. \cite{wilde13}) between joint and conditional entropies. We have also used $S(T)_{\sigma}\leq\log|\cT(\cX\times\cY,l)|\leq\dim(\cH_{A}\otimes\cH_{B})\log(l+1)$ and the marginal state
		\begin{align*}
		(\sigma_{s}^{XAB})^{\otimes l}&:=\sum_{x^{l}\in\cX^{l}}p^{ l}(x^{l})\ket{e_{x^{l}}}\bra{e_{x^{l}}}\otimes\sum_{\lambda}q_{x^{l}}(\lambda)\mathbbm{1}_{\cH_{A}^{\otimes l}}\otimes\cN_{s}^{\otimes l}(\Phi_{x^{l}}^{\lambda})\\&=\sum_{x^{l}\in\cX^{l}}p^{ l}(x^{l})\ket{e_{x^{l}}}\bra{e_{x^{l}}}\otimes\mathbbm{1}_{\cH_{A}^{\otimes l}}\otimes\cN_{s}^{\otimes l}(\Psi_{x^{l}})\\&=\omega_{s}^{\otimes l}(\cN_{s},p,\Psi),
		\end{align*}
		where $\Psi_{x^{l}}$ is a purification of $V^{\otimes l}(x^{l})$ and $\omega_{s}(\cN_{s},p,\Psi)$ is defined by (\ref{evaluationstate}). We use $\omega_{s}$ to denote this state. From (\ref{quantumrategen}) we have
		\begin{align}\label{generalqrate}
		\frac{1}{l\cdot a}\log M_{2,a\cdot l}&\geq\frac{1}{l}\inf_{s\in S}I(A\rangle B X,\omega_{s}^{\otimes l})-\frac{\dim(\cH_{A}\otimes \cH_{B})\log(l+1)}{l}-\frac{\delta}{l}\nonumber\\&=\inf_{s\in S}I(A\rangle B X,\omega_{s})-\frac{\dim(\cH_{A}\otimes\cH_{B})\log(l+1)}{l}-\frac{\delta}{l}.
		\end{align}
		Again from Lemma \ref{piinput} we have for $\delta>0$,
		\begin{align*}
		\frac{1}{a}\log M_{1,a\cdot l}&\geq\inf_{s\in S}I(A;B,\sigma_{s})-\delta\\&=\inf_{s\in S}S(\sum_{x^{l}}\sum_{\lambda} p^{ l}(x^{l})q_{x^{l}}(\lambda)\cN_{s}^{\otimes l}(\pi_{x^{l}}^{\lambda}))-\sum_{x^{l}}\sum_{\lambda} p^{l}(x^{l})q_{x^{l}}(\lambda)S(\cN_{s}^{\otimes l}(\pi_{x^{l}}^{\lambda}))-\delta\nonumber\\&\geq\inf_{s} S(\sum_{x^{l}}\sum_{\lambda} p^{ l}(x^{l})q_{x^{l}}(\lambda)\cN_{s}^{\otimes l}(\pi_{x^{l}}^{\lambda}))-\sum_{x^{l}} p^{ l}(x^{l})S(\cN_{s}^{\otimes l}(\sum_{\lambda}q_{x^{l}}(\lambda)\pi_{x^{l}}^{\lambda}))-\delta\\&=\inf_{s\in S} I(A;B,\omega_{s}^{\otimes l})-\delta
		\end{align*}
		and hence 
		\begin{align}
		\frac{1}{a\cdot l}\log M_{1,a\cdot l}\geq\inf_{s\in S}I(A;B,\omega_{s})-\frac{\delta}{l}.
		\end{align}
		
	For any block-length $n\in\mathbb{N}$, we can write $n=a\cdot l+r$ for $a,l,r\in\mathbb{N}$ and $0\leq r<l$. For all $0\leq r<l$, we use the above $(a\cdot l,M_{1,a\cdot l},M_{2,a\cdot l})$ CET codes to achieve the desired rate, observing that
	\begin{equation*}
	\liminf_{n\to\infty} \frac{1}{n}M_{i,n}\geq  \liminf_{a\to\infty} \frac{1}{a\cdot l}M_{i,a\cdot l}, \  i=1,2.
	 \end{equation*}
	 and that $P(\cC_{CET},\cN^{\otimes n},m)\geq P(\cC_{CET},\cN^{\otimes a\cdot l},m)$ for all $m\in[M_{a\cdot l}]$. 
	\end{proof}

	\begin{proof}[Proof of Lemma \ref{dp}]
		According to Lemma \ref{generalinput0}, 
		\begin{equation*}
		(R_{1},R_{2})\in\bigcup_{p,\Psi}\bigcap_{s\in S}\hat{C}(\mathcal{N}_{s},p,\Psi)
		\end{equation*}
		implies $(R_{1},R_{2})\in C_{CET}(\mathcal{J})$. Using standard double-blocking arguments, for each $l\in\mathbb{N}$, 
		\begin{equation*}
		(R_{1},R_{2})\in \bigcup_{l=1}^{\infty}\frac{1}{l}\bigcup_{p,\Psi}\bigcap_{s\in S}\hat{C}(\mathcal{N}_{s}^{\otimes l},p,\Psi)
		\end{equation*}
	implies	$(R_{1},R_{2})\in C_{CET}(\mathcal{J})$.
	\end{proof}	
	\section{Proofs for the arbitrarily varying quantum channel}\label{avqcsection}
	In this section we consider the task of simultaneous entanglement and classical message transmission in the AVQC model. We derive results for the CET capacities of such channels, when the uncertainty set generating the AVQC is finite. After proving the converse part in the following section, we have used Ahlswede's robustification and elimination techniques to derive suitable codes from compound codes developed so far to prove the direct part of the capacity theorem. Also we will remark the relevant positivity conditions based on results from \cite{advnoise}. 
	\subsection{Proof of converse }
	In this section, we prove the following lemma.
	\begin{lemma}
		Let $\cJ:=\{\cN_{s}\}_{s\in S}\subset\cC(\cH_{A},\cH_{B})$ with $|S|<\infty$ be an AVQC. We have 
		\begin{equation*}
		\overline{\cA}_{r,CET}(\cJ)\subset \overline{C}_{CET}(\conv(\cJ)).
		\end{equation*}
	\end{lemma}
	\begin{proof}
Let $(\mu_{l})_{l=1}^{\infty}$ be a sequence of random codes for AVQC generated by $\cJ$ with 
\begin{equation}\label{erroravc}
\lim_{l\to\infty}\inf_{s^{l}\in S^{l}}\int\frac{1}{M_{1,l}}\sum_{m\in[M_{1,l}]}g_{s^{l}}(\cP^{(m)},\cR^{(m)})\ d\mu_{l}(\cP^{(m)},\cR^{(m)})_{m\in[M_{1,l}]}=1
\end{equation}
with function $g_{s^{l}}$ defined by (\ref{fidfunc}) and $(\cP^{(m)},\cR^{(m)})_{m\in[M_{1,l}]}$ denoting the members of the singleton sets from the respective sigma-algebra. On the other hand, for the compound channel $\conv(\cJ)$ and each $\cN_{q}\in\conv(\cJ)$ we have
\begin{align*}
\int\frac{1}{M_{1,l}}\sum_{m\in[M_{1,l}]} F(\ket{m}\bra{m}&\otimes \Phi^{AB},\id_{\cH_{A}^{\otimes l}}\otimes \cR\circ\cN_{q}^{\otimes l}\circ\cP^{(m)}(\Phi^{AA}))\ d\mu_{l}(\cP^{(m)},\cR^{(m)})_{m\in [M_{1,l}]}=\\&\sum_{s^{l}\in S^{l}}q^{l}(s^{l})\int \frac{1}{M_{1,l}}\sum_{m\in[M_{1,l}]}g_{s^{l}}(\cP^{(m)},\cR^{(m)})\ d\mu_{l}(\cP^{(m)},\cR^{(m)})_{m\in [M_{1,l}]}\geq \\&\inf_{s^{l}\in S^{l}}\int \frac{1}{M_{1,l}}\sum_{m\in[M_{1,l}]}g_{s^{l}}(\cP^{(m)},\cR^{(m)})\ d\mu_{l}(\cP^{(m)},\cR^{(m)})_{m\in [M_{1,l}]}\geq 1-\epsilon_{l},
\end{align*}
with $\epsilon_{l}\searrow 0$. The last inequality comes from (\ref{erroravc}). This yields
\begin{align*}
\inf_{q\in\cP(S)}\int\frac{1}{M_{1,l}}&\sum_{m\in[M_{1,l}]} F(\ket{m}\bra{m}\otimes \Phi^{AB},\id_{\cH_{A}^{\otimes l}}\otimes \cR\circ\cN_{q}^{\otimes l}\circ\cP^{(m)}(\Phi^{AA}))d\mu_{l}(\cP^{(m)},\cR^{(m)})_{m\in [M_{1,l}]}\\&\geq 1-\epsilon_{l}.
\end{align*}
This means 
\begin{align*}
\int\frac{1}{M_{1,l}}&\sum_{m\in[M_{1,l}]} F(\ket{m}\bra{m}\otimes \Phi^{AB},\id_{\cH_{A}^{\otimes l}}\otimes \cR\circ\frac{1}{|\cP(S)|}\sum_{q\in\cP(S)}\cN_{q}^{\otimes l}\circ\cP^{(m)}(\Phi^{AA}))d\mu_{l}(\cP^{(m)},\cR^{(m)})_{m\in [M_{1,l}]}\\&\geq 1-\epsilon_{l},
\end{align*}
that in turn implies the existence of at least one CET code $(\cP^{(m)},\cR^{(m)})_{m\in [M_{1,l}]}$ for compound channel $\conv(\cJ)$ with average error lower-bounded by $1-|\cP(S)|\epsilon_{l}$. We therefore conclude  
	\begin{equation*}
	\overline{\cA}_{r,CET}(\cJ)\subset \overline{C}_{CET}(\conv(\cJ)).
	\end{equation*}
	\end{proof}
	\subsection{Proof of the direct part}
	In this section, we prove the following two lemmas, that along with the converse shown in the previous section, prove the first part of Theorem \ref{avcqmainresult}. 
	\begin{lemma}\label{robustification}
Let $\cJ:=\{\cN_{s}\}_{s\in S}\subset\cC(\cH_{A},\cH_{B})$ with $|S|<\infty$ be an AVQC. We have 
\begin{equation}
\overline{C}_{CET}(\conv(\cJ))\subset\overline{\cA}_{r,CET}(\cJ).
\end{equation}
	\end{lemma}
	\begin{lemma}\label{Theorem}
Let $\cJ:=\{\cN_{s}\}_{s\in S}\subset\cC(\cH_{A},\cH_{B})$ with $|S|<\infty$ be an AVQC. $\overline{\cA}_{d,CET}(\cJ)\neq\{(0,0)\}$ implies $\overline{\cA}_{d,CET}(\cJ)=\overline{\cA}_{r,CET}(\cJ)$. 
	\end{lemma}
	To prove the second part of Theorem \ref{avcqmainresult}, we invoke the following result from \cite{advnoise}. 
	\begin{theorem}(\cite{advnoise} Theorem 40)\label{m1=0}
		Let $\cJ=\{\cN_{s}\}_{s\in S}\subset\cC(\cH_{A},\cH_{B})$, $|S|<\infty$, be and AVQC. Then $\cJ$ is symmetrizable if and only if for all $\{\rho_{1},\dots,\rho_{M}\}\subset\cS(\cH_{A}^{\otimes l})$, $M,l\in\mathbb{N}$, $M\geq 2$, and POVMs $\{D_{m}\}_{m=1}^{M}$ on $\cH_{B}^{\otimes l}$, 
		\begin{equation*}
		\inf_{s^{l}\in S^{l}}\frac{1}{M}\sum_{m=1}^{M}(1-\tr(\cN_{s^{l}}(\rho_{m})D_{m}))\geq 1/4
		\end{equation*}
		holds. 
	\end{theorem}
	This result, along with the following lemma, prove the second part of Theorem \ref{avcqmainresult}. 
	\begin{lemma}\label{m1>m2}
		Let $(\cP,\cR)$ be an $(M,l)$ entanglement transmission code for AVQC $\cJ=\{\cN_{s}\}_{s\in S}\subset\cC(\cH_{A},\cH_{B})$ with
		\begin{equation}\label{quantumpart}
		F(\Phi^{AB},\id_{\cH_{A}^{\otimes l}}\otimes\cR\circ\cN_{s^{l}}\circ\cP(\Phi^{AA}))\geq 1-\epsilon\ \ \forall s^{l}\in S^{l}.
		\end{equation}
		Then, there exist $\{\rho_{1},\dots,\rho_{M}\}\subset\cS(\cH_{A}^{\otimes l})$ and POVM $\{D_{m}\}_{m\in[M]}$ on $\cH_{B}^{\otimes l}$ such that
		\begin{equation}\label{classicalpart}
		\frac{1}{M}\sum_{m=1}^{M}\tr(D_{m}\cN_{s^{l}}(\rho_{m}))\geq 1-\epsilon\ \ \forall s^{l}\in S^{l}
		\end{equation}
		holds. 
	\end{lemma}
	\begin{proof}
The proof follows directly from the convexity of entanglement fidelity in its first input and that 
\begin{equation*}
F(\Phi^{AB},\id_{\cH_{A}^{\otimes l}}\otimes\cR\circ\cN_{s^{l}}\circ\cP(\Phi^{AA}))=F_{e}(\pi_{\cF_{A,l}},\cR\circ\cN_{s^{l}}\circ\cP).
\end{equation*}
Defining for each $m\in[M]$, $D_{m}:=\cR_{*}(\ket{m}\bra{m})$ and $\rho_{m}:=\cP(\ket{m}\bra{m})$ with $\cR_{*}$ the Hilbert-Schmidt adjoint of channel $\cR$ and spectral decomposition $\pi_{\cF_{A,l}}=\frac{1}{M}\sum_{m\in[M]}\ket{m}\bra{m}$, we carry the lower bound on (\ref{quantumpart}) to (\ref{classicalpart}). 
		
	\end{proof}
	Lemma \ref{m1>m2} and Theorem \ref{m1=0} show that $\cJ$ is symmetrizable if and only if there exist no CET codes $(\cP^{(m)},\cR^{(m)})_{m\in M}$ with $M\geq 2$, such that we have $\inf_{s^{l}\in S^{l}}\frac{1}{M}\sum_{m\in[M]}g_{s^{l}}(\cP^{(m)},\cR^{(m)})\geq \frac{3}{4}$. This in turn implies the second part of Theorem \ref{avcqmainresult}. 
	\begin{proof}[Proof of Lemma \ref{robustification}]
		In Section \ref{directsection} (Lemma \ref{dp}, Lemma \ref{generalinput0}), it was shown that for large enough values of $l\in\mathbb{N}$ there exists CET codes $(\tilde{\cP}^{(m)},\tilde{\cR}^{(m)})_{m\in[M_{1,}]}$ of $(l,M_{1,l},M_{2,l})$ for compound channel $\conv(\cJ)$ that achieve the optimum capacity region of this channel $\overline{C}_{CET}(\conv(\cJ))$ with
		\begin{equation}\label{compounderror}
		\inf_{q\in\cP(S)}\frac{1}{M_{1,l}}\sum_{m\in[M_{1,l}]}F(\ket{m}\bra{m}\otimes\Phi^{AB},\id_{\cH_{A}^{\otimes l}}\otimes\tilde{R}\circ\cN_{q}^{\otimes l}\circ\cP^{(m)}(\Phi^{AA}))\geq 1-\epsilon_{l}
		\end{equation}
		with $\epsilon_{l}\to 0$ exponentially. Since 
		\begin{equation*}
		\cN_{q}^{\otimes l}=(\sum_{s\in S}q(s)\cN_{s})^{\otimes l}=\sum_{s^{l}\in S^{l}}q^{l}(s^{l})\cN_{s^{l}},
		\end{equation*}
		from (\ref{compounderror}) we obtain,
		\begin{equation}\label{avcerror}
\inf_{q\in\cP(S)}\sum_{s^{l}\in S^{l}}q^{l}(s^{l})\frac{1}{M_{1,l}}\sum_{m\in[M_{1,l}]}F(\ket{m}\bra{m}\otimes\Phi^{AB},\id_{\cH_{A}^{\otimes l}}\otimes\tilde{R}\circ\cN_{s^{l}}\circ\cP^{(m)}(\Phi^{AA}))\geq 1-\epsilon_{l}.
		\end{equation}
	Defining the function $f:S^{l}\to[0,1]$ with
	\begin{equation*}
	f(s^{l}):=\frac{1}{M_{1,l}}\sum_{m\in[M_{1,l}]}g_{s^{l}}(\tilde{\cP}^{(m)},\tilde{R}^{(m)}),
	\end{equation*}
	from (\ref{avcerror}) we obtain
	\begin{equation}
	\inf_{q\in\cP(S)}\sum_{s^{l}\in S^{l}}q^{l}(s^{l})f(s^{l})\geq 1-\epsilon_{l}. 
	\end{equation}
	Therefore the hypothesis of Ahlswede's robustification (Lemma \ref{robustlemma}) is satisfied and hence
	\begin{align}
	\frac{1}{l!}\sum_{\alpha\in \mathfrak{S}_{l}}\frac{1}{M_{1,l}}\sum_{m\in[M_{1,l}]}g_{s^{l}}(\cU_{A,\alpha}\circ\tilde{\cP}^{(m)},\tilde{R}^{(m)}\circ\cU_{B,\alpha}^{-1})\geq 1-(l+1)^{|S|}\epsilon_{l},
	\end{align}
	where $\cU_{X,\alpha}(\cdot)=U_{X,\alpha}(\cdot)U^{\dagger}_{X,\alpha}$ with $U_{X,\alpha}$ is a unitary on $\cH_{A}^{\otimes l}$, permuting the tensor factors on this Hilbert space according to $\alpha$, i.e.
	\begin{equation*}
	U_{X,\alpha}x_{1}\otimes\dots\otimes x_{l}=x_{\alpha(1)}\otimes\dots x_{\alpha(l)}.
	\end{equation*}
	Therefore the uniform distribution over the set $\{(\cP_{\alpha}^{(m)},\cR_{\alpha}^{(m)})_{m\in[M_{1,l}]}:\alpha\in\mathfrak{S}_{l}\}$ with 
	\begin{equation*}
	\cP_{\alpha}^{(m)}:=\cU_{A,\alpha}\circ\tilde{\cP}^{(m)}
	\end{equation*}
	and 
	\begin{equation*}
\cR_{\alpha}^{(m)}:=\tilde{\cR}^{(m)}\circ\cU_{B,\alpha}^{-1},
	\end{equation*}
		yield the desired random CET code for arbitrarily varying channel generated by $\cJ$. Hence we conclude that $(R_{1},R_{2})\in\overline{C}_{CET}(\conv(\cJ))$ implies $(R_{1},R_{2})\in\overline{\cA}_{r,CET}(\cJ)$.
	\end{proof}
	To prove Lemma \ref{Theorem}, we need the following statement.
	\begin{lemma}\label{eliminationlemma}
		Let $\cJ:=\{\cN_{s}\}_{s\in S}$ with $|S|<\infty$ be an AVQC, $l\in\mathbb{N}$, $\mu_{l}$ an $(l,M_{1,l},M_{2,l})$ random CET code for $\cJ$ with
		\begin{equation}\label{hypothesisint}
		\inf_{s^{l}\in S^{l}}\int\frac{1}{M_{1,l}}\sum_{m\in[M_{1,l}]}g_{s^{l}}(\cP^{(m)},\cR^{(m)}) d\mu_{l}(\cP^{(m)},\cR^{(m)})_{m\in[M_{1,l}]}\geq 1-\epsilon_{l}
		\end{equation}
		for a sequence $(\epsilon_{l})_{l\in\mathbb{N}}$ such that $\epsilon_{l}\searrow 0$. Then, for $\epsilon\in(0,1)$ and sufficiently large $l\in\mathbb{N}$, there exist $l^{2}$ $(l,M_{1,l},M_{2,l})$ CET codes $\{(\cP_{i}^{(m)},\cR_{i}^{(m)})_{m\in[M_{1,l}]}\}$ with
		\begin{equation*}
		\frac{1}{l^{2}}\sum_{i=1}^{l^{2}}\frac{1}{M_{1,l}}\sum_{m\in[M_{1,l}]}g_{s^{l}}(\cP^{(m)}_{i},\cR^{(m)}_{i})\geq 1-\epsilon\ (\forall s^{l}\in S^{l}).
		\end{equation*}
		\end{lemma}
		\begin{proof}
			Let for $K\in\mathbb{N}$, $(\Lambda_{i}^{(m)},\Gamma_{i}^{(m)})_{m\in[M_{1,n}]}$ for $i=1,\dots,K$ be independent random variables with values in $\cC(\cF_{A,l},\cH_{A}^{\otimes l})^{M_{1,l}}\times\Omega_{l}$ distributed according to $\mu_{l}^{\otimes K}$. We use the shorthand notation
			\begin{equation*}
h_{s^{l}}(i):=\frac{1}{M_{1,l}}\sum_{m\in[M_{1,l}]}g_{s^{l}}(\Lambda_{i}^{(m)},\Gamma_{i}^{(m)}).
			\end{equation*}
			For every $s^{l}\in S^{l}$, an application of Markov's inequality for every $\epsilon\in (0,1)$ and $\gamma>0$ yields
			\begin{align}\label{Markovineq}
			\mathbb{P}[1-\frac{1}{K}\sum_{i=1}^{K}h_{s^{l}}(i)\geq\epsilon/2]=\mathbb{P}[2^{K\gamma-\gamma\sum_{i=1}^{K}h_{s^{l}}(i)}\geq 2^{K\gamma(\epsilon/2)}]\leq 2^{-K\gamma(\epsilon/2)}\mathbb{E}[2^{\gamma(K-\sum_{i=1}^{K}h_{s^{l}}(i))}].
			\end{align}
			We now upper-bound the expectation in (\ref{Markovineq}).
			\begin{align}\label{56}
\mathbb{E}[2^{\gamma(K-\sum_{i=1}^{K}h_{s^{l}}(i))}]=(\mathbb{E}[2^{\gamma(1-h_{s^{l}}(1))}])^{K}\leq(\mathbb{E}[(1+2^{\gamma}(1-h_{s^{l}}(1)))])^{K}\leq (1+2^{\gamma}\epsilon_{l})^{K}.
			\end{align}
			The second inequality is due to the fact that $ (\Lambda_{i}^{(m)},\Gamma_{i}^{(m)})_{m\in[M_{1,n}]}$ are i.i.d for $i=1,\dots,K$, the first inequality comes from $2^{\gamma t}\leq(1-t)2^{0\cdot\gamma}+t2^{\gamma}\leq 1+2^{\gamma t}$ for $t\in[0,1]$ and last inequality comes from (\ref{hypothesisint}). For $K=l^{2}$ and $\gamma=2$ therefore, there exists $l_{0}(\epsilon)\in\mathbb{N}$ such that for $l\geq l_{0}(\epsilon)$
			\begin{equation}\label{l0}
			(1+2^{\gamma}\epsilon_{l})^{l^{2}}\leq 2^{l^{2}(\epsilon/2)}.
			\end{equation}
			Therefore we obtain from (\ref{Markovineq}), (\ref{56}) and (\ref{l0}),
			\begin{equation*}
			\mathbb{P}[1-\frac{1}{l^{2}}\sum_{i=1}^{l^{2}}h_{s^{l}}(i)\geq\epsilon/2]\leq 2^{-l^{2}(\epsilon/2)}.
			\end{equation*}
			Applying the union bound on the last inequality yields
			\begin{equation*}
	\mathbb{P}[\frac{1}{l^{2}}\sum_{i=1}^{l^{2}}h_{s^{l}}(i)>1-\epsilon/2, \forall s^{l}\in S^{l}]\geq 1-|S|^{l} 2^{-l^{2}(\epsilon/2)},
			\end{equation*}
				which implies that there is a realization $(\cP_{i}^{(m)},\cR_{i}^{(m)})_{m\in[M_{1,l}]}, i=1,\dots l^{2}$ with
				\begin{equation*}
\frac{1}{l^{2}}\sum_{i=1}^{l^{2}}\frac{1}{M_{1,l}}\sum_{m\in[M_{1,l}}g_{s^{l}}(\cP_{i}^{(m)},\cR_{i}^{(m)})>1-\epsilon/2 \ \ \forall s^{l}\in S^{l},
				\end{equation*}
				when $|S|^{l}2^{-l^{2}(\epsilon/2)}<1$ which is possible for sufficiently large values of $l$. 
		\end{proof}
	\begin{proof}[Proof of Lemma \ref{Theorem}]
		By assumption, for $\epsilon\in(0,1)$ there exists a $(r_{l},l^{2},1)$ deterministic CET code $(\tilde{\cP}^{(m)},\tilde{\cR}^{(m)})_{m=1}^{l^{2}}$ with 
		\begin{equation}\label{l2crit}
\frac{1}{l^{2}}\sum_{m=1}^{l^{2}}g_{s^{r_{l}}}(\tilde{\cP}^{(m)},\tilde{\cR}^{(m)})\geq 1-\epsilon \ \forall s^{r_{l}}\in S^{r_{l}},
		\end{equation}
		with $r_{l}=o(l)$. This is because if the capacity region is not equal to the point $(0,0)$, $R_{1}$ (intersection of the capacity region with the $x$-axis), is definitely larger than zero (see Lemma \ref{m1>m2}). On the other hand, let $(R_{1},R_{2})\in\overline{\cA}_{r,CET}$. By Lemma \ref{eliminationlemma}, this implies the existence of $l^{2}$ $(l,M_{1,l},M_{2,l})$ CET codes $\{\hat{\cP}^{(m)}_{i},\hat{\cR}^{(m)}_{i}:i\in[l^{2}]\}$ of the same rate with 
		\begin{equation}\label{maincodes}
		\frac{1}{l^{2}}\sum_{i=1}^{l^{2}}\frac{1}{M_{1,l}}\sum_{m=1}^{M_{1,l}}g_{s^{l}}(\hat{\cP}^{(m)}_{i},\hat{\cR}^{(m)}_{i})\geq 1-\epsilon\ \forall s^{l}\in S^{l}.
		\end{equation}
		Define CPTP maps
		\begin{equation*}
		\cP^{(m)}(a\otimes b):=\frac{1}{l^{2}}\sum_{i=1}^{l^{2}}\tilde{\cP}^{(i)}(a)\otimes\hat{\cP}^{(m)}_{i}(b),
		\end{equation*}
		\begin{equation*}
		\cR^{(m)}(c\otimes d):=\sum_{i=1}^{l^{2}}\tilde{\cR}^{(i)}(c)\otimes\hat{\cR}^{(m)}_{i}(d).
		\end{equation*}
		We have
		\begin{align}\label{combocodes}
		\frac{1}{M_{1,l}}\sum_{m=1}^{M_{1,l}}g_{s^{r_{l}+l}}(\cP^{(m)},\cR^{(m)})&=\nonumber\\&\frac{1}{l^{2}\cdot M_{1,l}}\sum_{m=1}^{M_{1,l}}F(\Phi^{AB},\id_{\cH_{A}^{\otimes r_{l}+l}}\otimes\sum_{i=1}^{l^{2}}\tilde{\cR}^{(i)}\otimes\hat{\cR}^{(m)}_{i}\circ\cN_{s^{r_{l}}}\otimes\cN_{s^{l}}\circ\sum_{j=1}^{l^{2}}\tilde{\cP}^{(j)}\otimes\hat{\cP}^{(m)}_{i}(\Phi^{AA}))\nonumber\\&\geq \frac{1}{l^{2}}\sum_{i=1}^{l^{2}}\frac{1}{ M_{1,l}}\sum_{m=1}^{M_{1,l}}F(\tilde{\Phi}^{AB},\id_{\cH_{A}^{\otimes r_{l}}}\otimes\tilde{\cR}^{(i)}\circ\cN_{s^{r_{l}}}\circ\tilde{\cP}^{(j)}(\tilde{\Phi}^{AA}))\nonumber\\&\times F(\hat{\Phi}^{AB},\id_{\cH_{A}^{\otimes l}}\otimes\hat{\cR}^{(m)}_{i}\circ\otimes\cN_{s^{l}}\circ\hat{\cP}^{(m)}_{i}(\hat{\Phi}^{AA})),
		\end{align}
		where $\tilde{\Phi}^{XY}$ and $\hat{\Phi}^{XY}$ are maximally entangled states. The inequality above is due to the fact that $g_{s^{r_{l}+l}}(\cP^{(m)},\cR^{(m)})$ is non-negative for all $m$ and $s^{l+r_{l}}$. Applying Lemma \ref{2epsilonlemma} on (\ref{combocodes}), given (\ref{l2crit}) and (\ref{maincodes}) we conclude
		\begin{equation*}
\frac{1}{M_{1,l}}\sum_{m=1}^{M_{1,l}}g_{s^{r_{l}+l}}(\cP^{(m)},\cR^{(m)})\geq 1-2\epsilon.
		\end{equation*}
		As $r_{l}=o(l)$, this implies $(R_{1},R_{2})\in\overline{\cA}_{d,CET}(\cJ)$. This in turn implies $\overline{\cA}_{r,CET}(\cJ)\subset\overline{\cA}_{d,CET}(\cJ)$. As the inclusion $\overline{\cA}_{d,CET}(\cJ)\subset\overline{\cA}_{r,CET}(\cJ)$ is obvious, we are done. 
	\end{proof}
	\section{Simultaneous classical message and entanglement transmission over fully quantum AVCs}\label{classicallyenhanced}
		In this section, we consider simultaneous transmission of classical messages and entanglement over an an arbitrarily varying quantum channel with a \emph{quantum jammer}.
		Let $\cN \in \cC(\cH_A \otimes \cH_J, \cH_B)$ be a quantum channel whose input space is a tensor product of a Hilbert space $\cH_A$ (the legitimate sender's space) and a Hilbert space $\cH_J$ which is under control of a quantum jammer. We consider a situation, where for each given block-length $n$, the jammer may choose any state $\eta$ on $\cH_J^{\otimes n}$ as input in order to disturb the transmission of the legitimate parties. \newline
		The \emph{Arbitrarily Varying Quantum Channel (AVQC)} generated by $\cN$ is given by the family
		\begin{align}
		\left\{\cN_{n,\sigma}(\cdot) := \cN^{\otimes n}(\cdot \otimes \sigma): \sigma \in \cS(\cH_J^{\otimes n}), n \in \bbmN \right\} \label{full_avqc_channel_description}
		\end{align}
		of CPTP maps\footnote{Although acronym "AVQC" is also used for the somewhat more restrictive channel model introduced in Section \ref{codedefinitionavcq}, it should be apparent from context, which of these models is considered.}. 
		The above channel model already has been under consideration in case of univariate transmission goals. Karumanchi et al. \cite{manchini16} utilized the postselection technique from \cite{christandl09a} to derive correlated random codes for the AVQC from good codes for the compound channel generated by $\fI := \{\cN_{\sigma}:= \cN(\cdot, \sigma) : \ \sigma \in \cS(\cH_J)\}$. This approach turned out to be successful to determine the random entanglement transmission capacity for the AVQC. In recent work  \cite{boche18}, the above mentioned techniques were used to also characterize the random classical message transmission capacity of the AVQC. Going beyond, the authors of \cite{boche18} introduced a derandomization technique to derive a dichotomy for the entanglement and classical message transmission capacities of the QAVC. 
		\emph{The deterministic capacity is zero or it equals the random capacity.} 
		We show, that the ideas of the mentioned works together with the results derived in this paper are sufficient to determine the random capacity and establish a partial characterization of the deterministic capacity in terms of a dichotomy also in case of simultaneous transmission of entanglement and classical messages. \newline 
		The definitions for the corresponding capacity regions can be easily extrapolated from the corresponding definitions in Section \ref{codedefinitionavcq} using the set of transmission maps in (\ref{full_avqc_channel_description}). We denote the \emph{random CET capacity region} of $\cN$ by $\overline{\cA}_{r,CET}(\cN)$ and the \emph{deterministic CET capacity} by $\overline{\cA}_{d,CET}(\cN)$. First, we give a characterization of the random CET capacity $\overline{\cA}_{r,CET}(\cN)$ of the AVQC with fully quantum jammer. 
		\begin{theorem} \label{thm:full_avqc_rand_cap}
			Let $\cN \in \cC(\cH_A \otimes \cH_J,  \cH_B)$, and $\fI := \{\cN_{\sigma}: \ \sigma \in \cS(\cH_J) \}$. It  holds 
			\begin{align}
			\ \overline{\cA}_{r,CET}(\cN) \ = \ \overline{C}_{CET}(\fI) \  \label{thm:full_avqc_rand_cap_1}
			\end{align}
		\end{theorem}
		The $\supset$ inclusion in (\ref{thm:full_avqc_rand_cap}) is obvious. To show the reverse inclusion, we will invoke the " robustification " statement in Proposition \ref{prop:full_qantum_robustification} below. In the derivations, the following representation of the permutation group $\mathfrak{S}_n$ on $n$-fold tensor product spaces plays a key role. Let for each $\pi \in \mathfrak{S}_n$, $U_\pi$ be the unitary exchanging the factors in $\cH^{\otimes n}$, i.e. 
		\begin{align*}
		U_\pi \ x_1 \otimes \cdots \otimes x_n = x_{\pi(1)} \otimes \dots \otimes x_{\pi(n)}
		\end{align*}
		for each $x_1,\dots,x_n \in \cH$. We set $\cU_\pi(\cdot) := U_\pi (\cdot) U^\ast_\pi$. In $\cU_{A,\pi}, \ \cU_{B,\pi}, \cU_{J,\pi}$ denote the corresponding maps performed on the subsystems under control of $A,B,J$ accordingly. A rather powerful result for states being invariant under permutations of the tensor factors is the following. 
		\begin{proposition}[de Finetti reduction \cite{christandl09a}] \label{prop:post-selection}
			Let $\rho \in \cS(\cH^{\otimes n})$ permutation invariant, i.e. $\cU_\pi(\rho) = \rho$ for each $\pi \in \mathfrak{S}_n$. It holds
			\begin{align*}
			\rho \ \leq \ (n+1)^{(\dim \cH)^2} \ \int \sigma^{\otimes n} d\mu(\sigma)
			\end{align*}
			with a probability measure $\mu$. 
		\end{proposition}
		\begin{proposition}\label{prop:full_qantum_robustification}
			Let $\cC:=(\cP_m, \cR_m)_{m=1}^{M_1}$ be an $(n,M_1,M_2)$-CET code such that with $\lambda \in (0,1)$ 
			\begin{align*}
			\underset{\sigma \in \cS(\cH_J)}{\inf} \ \overline{P}_{CET}(\cC, \cN_{\sigma }^{\otimes n})) \ \geq 1 - \lambda
			\end{align*}
			holds. With $\cC_{\pi} := (\cU_{A,\pi}\circ \cP_m, \cR_m \circ \cU_{B,\pi^{-1}})$ for each $\pi \in \mathfrak{S}_n$, it holds
			\begin{align*}
			\underset{\tau \in \cS(\cH_J^{\otimes n})}{\inf} \ \frac{1}{n!} \sum_{\pi \in \mathfrak{S}_n} \ \overline{P}_{CET}(\cC_\pi, \cN_{n,\tau}) \ \geq 1-(n+1)^{(\dim \cH_J)^2} \cdot \lambda.
			\end{align*}
		\end{proposition} 
		\begin{proof}
			The proof closely follows the lines of \cite{manchini16}. Set $d_J := \dim \cH_J$. By permutation invariance of $\cN^{\otimes n}$, the equality 
			\begin{align}
			\cU_{B,\pi^{-1}} \circ \cN^{\otimes n} \circ (\cU_{A,\pi} \otimes \id_{\cH_J}^{\otimes n}) = \cN^{\otimes n}\circ (\id_{\cH_A}^{\otimes n} \otimes \cU_{J,\pi^{-1}}) \label{prop:full_quantum_robustification_1}
			\end{align}
			holds for each permutation $\pi \in \mathfrak{S}_n$. Using (\ref{prop:full_quantum_robustification_1}) together with the fact, that $\overline{P}_{CET}$ is an affine function of the channel, we obtain 
			\begin{align*}
			\frac{1}{n!} \sum_{\pi \in \mathfrak{S}_n} \ \overline{P}_{CET}(\cC_\pi, \cN_{n,\tau}) \ = \ \overline{P}_{CET}(\cC, \cN_{n,\overline{\tau}})
			\end{align*} 
			for each $\tau \in \cS(\cH_J^{\otimes n})$, where $\overline{\tau} := 1/n!\ \sum_{\pi \in \mathfrak{S}_n} \cU_{\pi}(\tau)$. Define 
			$\cT_\ast$ to be the Hilbert-Schmidt adjoint of the map
			\begin{align*}
			\sigma \ \mapsto \ \frac{1}{M_1} \sum_{m=1}^{M_1}\id \otimes \cR_m \circ \cN_{n,\sigma} \circ \cP_m(\Phi). 
			\end{align*}
			We write 
			\begin{align}
			1 - \overline{P}_{CET}(\cC, \cN_{n,\tau}) \ = \ \tr X \tau, \label{prop:full_quantum_robustification_2} 
			\end{align}
			with the matrix $X := \bbmeins - \cT_\ast(\Phi)$ (note that $0 \leq X \leq \bbmeins$ holds.) Using Proposition \ref{prop:post-selection} together with linearity and monotonicity of the integral, we have
			\begin{align*}
			\tr X \overline{\tau} \ 
			& \leq \ (n+1)^{d_J^2} \int \tr X\sigma^{\otimes n} \ d\mu(\sigma) \\
			& \leq  (n+1)^{d_J^2} \underset{\sigma \in \cS(\cH)}{\sup} \tr X\sigma^{\otimes n}. \\
			& \leq  (n+1)^{d_J^2} \cdot \lambda.
			\end{align*}
			Which is, by (\ref{prop:full_quantum_robustification_2}), the desired bound. 
		\end{proof} 
		\begin{proof}[Proof of Theorem \ref{thm:full_avqc_rand_cap} (Direct part)]
			The statement $\overline{C}_{CET}(\fI) \ \subset \ \overline{\cA}_{r,CET}(\cN)$ directly follows from combining the results from Section \ref{directsection} (Lemma \ref{dp} and Lemma \ref{generalinput0}) with Proposition \ref{prop:full_qantum_robustification}. Let $(\cC_n)_{n=1}^{\infty}$ be a 
			sequence of $(n, M_{1,n},M_{2,n})$-CET codes with 
			\begin{align*}
			\underset{\sigma \in \cS(\cH_J)}{\inf} \ \overline{P}_{CET}(\cC_n, \cN_\sigma^{\otimes n}) \geq 1 - 2^{-n c}
			\end{align*}
			with a constant $c > 0$	for each large enough $n$. Let $\tilde{\mu}_n$ be the uniform distribution on $\mathfrak{S}_n$, and $f(\pi) := \cC_{n,\pi}$. Then $\tilde{\mu}_n \circ f^{-1}$ is an $(n,M_{1,n}, M_{2,n})$ random CET code, such that 
			\begin{align*}
			\bbmE \left[\underset{\sigma \in \cS(\cH_J)}{\inf}\overline{P}_{CET}(\cdot, \cN_{n,\sigma})\right] \ \geq 1 - (n+1)^{d_J^2} \cdot 2^{-nc}.
			\end{align*}
			Since the right hand side of the above inequality tends to one for $n \rightarrow \infty$, every rate pair $(R_1,R_2)$ being achievable for the compound channel $\fI$, is also achievable by random codes for the AVQC $\cN$.
		\end{proof}
		Next we show, using a derandomization technique introduced in \cite{boche18}, the following statement. 
		\begin{theorem}[Dichotomy for $\overline{\cA}_{d,CET}$]\label{thm:full_avqc_dichotomy}
			$\overline{\cA}_{d,CET}(\cN)$ equals $\{(0,0)\}$ or $\overline{\cA}_{r,CET}(\cN)$
		\end{theorem}
		\begin{remark}
			The above statement quantifies the deterministic capacity region of the AVQC up to a blind spot. It is an open question whether or not there are channels for which $\overline{\cA}_{d,CET}(\cN) = \{(0,0)\}$ and $\{(0,0) \} \subsetneq \overline{\cA}_{r, CET}(\cN)$ does happen. 
		\end{remark}
		\begin{fact} \label{fact:trace_loewner}
			For a Hermitian matrix $A \in \cL(\cH)$, and $\alpha > 0$, it holds $A \leq \alpha \bbmeins$ if and only if $\tr \sigma A \ \leq 
			\alpha$ for all $\sigma \in \ \cS(\cH)$.
		\end{fact}
		\begin{proposition}[\cite{ahlswede02}, Theorem 19] \label{prop:matrix_chernoff}
			Let $X_1,\dots, X_T$ be i.i.d. hermitian random matrices with 
			$0 \leq X_i \leq \bbmeins$ a.s. for all $i \in [T]$, and $\bbmE X_1 \ \leq \ m \bbmeins \leq a \bbmeins \leq A$. Then 
			\begin{align*}
			\mathbb{P}\left(\frac{1}{T} \sum_{t=1}^T X_t \ \geq a \bbmeins_{\cH} \right) \ \leq \ \dim \cH \cdot \exp(- T 2 (a-m)^2)
			\end{align*}
		\end{proposition}
		\begin{proof}[Proof of Theorem \ref{thm:full_avqc_dichotomy}]
			We consider the non-trivial case $\overline{\cA}_{d,CET} \ \neq \{(0,0)\}$. \newline Let $(R_1, R_2) \in \overline{\cA}_{r,CET}(\cN) \setminus \{(0,0)\}$. We aim to show that $(R_1,R_2)$ is also achievable with deterministic codes. Since $\overline{\cA}_{d,CET} \neq \{(0,0)\}$, we find, for each large enough blocklength $n$ an $(n,\tilde{M}_1,\tilde{M}_2)$-CET code $\cC^{(1)} := (\cP^{(1)}_m, \cR^{(1)}_m)_{m=1}^{\tilde{M}_1}$ with $\tilde{M}_1 \geq 2^{l \tilde{R}}$, where $\tilde{R}>0$ is a constant, and 
			\begin{align}
			\underset{\sigma \in \cS(\cH_J^{\otimes l})}{\inf}\overline{P}_{CET}(\cC^{(1)}, \cN_{n,\sigma}) \geq 1 - \epsilon_l
			\end{align}
			with $\epsilon_n \rightarrow 0$ for $n \rightarrow \infty$. Set for each $n$, $a_n := \lceil 2\log n / \tilde{R} \rceil$, and $b_n := n - a_n$, i.e. $n = a_n + b_n$. If $n$ is large enough, 
			we have a random $(n,M_1, M_2)$-CET code  $\mu_{b_n}$ such that 
			$\frac{1}{b_n}\log M_i \geq R_i - \delta$, for $i = 1,2$, and 
			\begin{align}
			\bbmE_{\mu_{b_n}} \underset{\sigma \in \cS(\cH_J^{\otimes b_n})}{\inf} \overline{P}_{CET}(\cdot, \cN_{n,\sigma}) \ \geq 1 - 2^{-b_nc}. \label{thm:full_avqc_dichotomy_2} 
			\end{align}
			For simplicity, we assume $\mu_{b_n}$ to be finitely supported on $\{\cC_1,\dots,\cC_{T'}\}$ (which is possible by the explicit construction of a finite random code in Proposition \ref{prop:full_qantum_robustification}). Note, that we can write 
			\begin{align}
			1 - \overline{P}_{CET}(\cC_t, \cN_{b_n,\sigma}) \ = \ \tr E_t \sigma \label{thm:full_avqc_dichotomy_3}
			\end{align}
			with a matrix $0\leq E_t\leq\mathbbm{1}$ for each $t \in [T']$. By (\ref{thm:full_avqc_dichotomy_2}) and (\ref{thm:full_avqc_dichotomy_3}), together with linearity of expectation, it holds
			\begin{align}
			\underset{\sigma \in \cH_J^{\otimes b_n}}{\inf}\tr \overline{E}\sigma \
			= \ \bbmE_{\mu_{b_n}}\left[ 1 - \underset{\sigma \in \cH_J^{\otimes b_n}}{\inf} \overline{P}_{CET}(\cdot, \cN_{n,\sigma}) \right]  \ 
			\leq 2^{-b_nc}, \label{thm:full_avqc_dichotomy_3_a}
			\end{align}
			where we defined $\overline{E} := \bbmE_{\mu_{b_n}} E_t$. By Fact \ref{fact:trace_loewner}, combined with the bound in (\ref{thm:full_avqc_dichotomy_3_a}), $\overline{E} \leq 2^{-nc}\bbmeins$. Let $X^{(1)},\dots X^{(\tilde{M}_1)}$ be i.i.d. random matrices, each distributed according to $\mu_{b_n}$. By Proposition \ref{prop:matrix_chernoff}, It holds
			\begin{align}
			\mathbb{P} \left(\frac{1}{\tilde{M}_1}\sum_{t=1}^{\tilde{M}_1} E_{t} \geq \left(2^{-nc} + \frac{1}{n}\right) \bbmeins \right) \ \leq \ d_J^{b_n} \exp \left(- \tilde{M}_1/n^2 \right). \label{thm:full_avqc_dichotomy_4}
			\end{align}
			By our choice of $a_n$, the RHS of (\ref{thm:full_avqc_dichotomy_4}) is strictly smaller than one for each large enough $n$. Therefore, we find $\cC_1,\dots, \cC_{\tilde{M}_1}$ such that 
			\begin{align*}
			1 - \frac{1}{\tilde{M}_1} \sum_{t=1}^{\tilde{M}_1}\overline{P}_{CET}(\cC_t, \cN_{n,\sigma}) \ \leq \ 2^{-b_n c} + \frac{1}{n} := \gamma_n
			\end{align*}
		holds. Let $\cC_t = (\cP_{t,m}^{(2)}, \cR_{t,m}^{(2)})_{m=1}^{M_1}$. We define an $(n,M_1,M_2)$ deterministic CET code $\cC = (\cP_m, \cR_m)_{m=1}^{M_1}$ with
			\begin{align*}
			\cP_m := \frac{1}{\tilde{M}_1} \sum_{t=1}^{\tilde{M}_1}  \cP_t^{(1)}(\pi_1) \otimes \cP_{t,m}^{(2)}, \hspace{.5cm} \text{and} \hspace{1cm} \cR_m \ := \ \tr_{\cH^{\otimes a_n}}\circ \sum_{t=1}^{\tilde{M_1}} \cR^{(1)}_{t} \otimes \cR^{(2)}_{t,m}
			\end{align*}
			To evaluate the fidelity of the above code, we notice, that for each $\sigma \in \cS(\cH_J^{\otimes n}), t \in [\tilde{M}_1], m \in [M_1]$
			\begin{align}
				F(\Phi_1 \otimes \Phi_2, \id \otimes \circ \cR_{t}^{(1)} \otimes \cR_{t,m}^{(2)} \circ \cN_{n,\sigma}\circ \cP_t^{(1)} \otimes
				\cP_{t,m}^{(2)}(\Phi_1 \otimes \Phi_2)) = \tr F_t^{(1)} \otimes F_{t,m}^{(2)} \sigma \label{av_gleichung}
			\end{align}
			holds with effects $F^{(1)}_t$, $F^{(2)}_{t,m}$. This is advantageous, since    \begin{align*}
				\frac{1}{\tilde{M}_1}\sum_{t=1}^{\tilde{M}_1}\tr F^{(1)}_{t} \tau \ = \ \overline{P}_{CET}(\cC^{(1)}, \cN_{a_n, \tau}) \ \geq \ 1 - \epsilon_n,
			\end{align*}
			and
			\begin{align*}
				\frac{1}{M_1}\frac{1}{\tilde{M}_1}\sum_{t=1}^{\tilde{M}_1}\sum_{m=1}^{M_1}\tr
				F^{(2)}_{t,m} \tau \ = \ \overline{P}_{CET}(\cC_t^{(2)}, \cN_{b_n,
					\tau})  \ \geq 1 - \gamma_n.
			\end{align*}
			We have for each $\sigma \in \cS(\cH_J^{\otimes n})$
			\begin{align}
				&\overline{P}_{CET}(\cC, \cN_{n,\sigma}) \nonumber \\
				&= \frac{1}{M_1 \tilde{M}_1} \sum_{t,t' =1}^{\tilde{M}_1} \sum_{m=1}^{M_1}
				F\left(\Phi_2, \id_{\cH_J^{\otimes b_n}} \otimes \tr_{\cH_J^{\otimes a_n}} \circ \cR^{(1)}_{t'}\otimes \cR^{(2)}_{t',m} \circ \cN_{n,\sigma} \circ \cP_t^{(1)}(\pi_1)\otimes
				\cP_{t,m}^{(2)}(\Phi_2)\right) \nonumber \\
				&\geq \frac{1}{M_1 \tilde{M}_1} \sum_{t,t' =1}^{\tilde{M}_1} \sum_{m=1}^{M_1}
				F\left(\Phi_1 \otimes \Phi_2, \id_{\cH_J^{\otimes b_n}} \otimes  \cR^{(1)}_{t'}\otimes \cR^{(2)}_{t',m} \circ \cN_{n,\sigma} \circ \cP_t^{(1)}\otimes \cP_{t,m}^{(2)}(\Phi_1 \otimes \Phi_2)\right) \nonumber  \\
				&\geq \frac{1}{M_1 \tilde{M}_1} \sum_{t =1}^{\tilde{M}_1} \sum_{m=1}^{M_1}
				F\left(\Phi_1 \otimes \Phi_2, \id_{\cH_J^{\otimes b_n}} \otimes  \cR^{(1)}_{t}\otimes \cR^{(2)}_{t,m} \circ \cN_{n,\sigma} \circ \cP_t^{(1)}\otimes \cP_{t,m}^{(2)}(\Phi_1 \otimes \Phi_2)\right) \nonumber \\
				& = \frac{1}{M_1 \tilde{M}_1} \sum_{t =1}^{\tilde{M}_1} \sum_{m=1}^{M_1}
				\tr F^{(1)}_t \otimes F^{(2)}_{t,m} \sigma. 
				\label{thm:full_avqc_dichotomy_5}
			\end{align}
			The first inequality above is by monotonicity of the fidelity under CPTP maps. The last equality is from (\ref{av_gleichung}). Now, let $\sigma_1$ be the marginal of $\sigma$ on the first $a_n$ tensor factors of $\cH_J^{\otimes n}$, and $\sigma_2$ the marginal on the last $b_n$ tensor factors. 	\begin{align*} 
			A \otimes B  \ \geq \  \bbmeins \otimes \bbmeins - \bbmeins \otimes (\bbmeins - B) - (\bbmeins - A) \otimes \bbmeins 
			\end{align*}
			which holds for any two matrices $0 \leq A,B \leq \bbmeins$. We have
			\begin{align}
			\tr F^{(1)}_t \otimes F^{(2)}_{t,m} \sigma \ \geq 1 - \ \tr  (\bbmeins - F^{(1)}_t)   \sigma_1 -  \tr (\bbmeins - F^{(2)}_{t,m}) \sigma_2 \label{thm:full_avqc_dichotomy_6}
			\end{align}
			Combining (\ref{thm:full_avqc_dichotomy_5}), and (\ref{thm:full_avqc_dichotomy_6}), we can bound
			\begin{align*}
			\overline{P}_{SET}(\cC, \cN_{n,\sigma}) \ 
			& \geq  \ P(\cC^{(1)}, \cN_{a_n,\sigma_1}) +  \frac{1}{\tilde{M}_1} \ \sum_{t=1}^{\tilde{M}1}P(\cC^{(2)}_t, \cN_{b_n,\sigma_2}) - 1 
			\end{align*}
			Minimizing over all states on $\cH_J^{\otimes n}$, we obtain
			\begin{align*}
			\underset{\sigma \in \cS(\cH_J^{\otimes n})}{\inf} \overline{P}_{SET}(\cC, \cN_{n,\sigma}) \ \geq 1 - \epsilon_n - \gamma_n.
			\end{align*}
			The right hand side approaches one for $n \rightarrow \infty$. Since also $\frac{a_n}{n} \rightarrow 0$ and $\frac{b_n}{n} \rightarrow 1$ for $n \rightarrow \infty$, it is clear, that we achieve $(R_1,R_2)$ with the codes defined.
			
		\end{proof}
			
	\section{Concluding remarks and future work}
	We have developed universal codes for simultaneous transmission of classical information and entanglement under possible jamming attacks by a third malignant party. In the compound channel model, the quantum part of information transmission was done under two important scenarios of entanglement transmission and entanglement generation. The present random codes differ from those used for the perfectly known channel in \cite{idev}. We therefore did not need to approximate our input random codes by an i.i.d state (one with tensor product structure). Also, we evaded BSST type lemmas used in \cite{boche17} by using basic concavity properties of von Neumann entropy. Future work will hopefully include another important scenario under which quantum information is transmitted, namely subspace transmission. An equivalence statement between the strong subspace transmission and entanglement transmission has been proven in \cite{advnoise}. Recently in \cite{gisbert}, an instance of the present classically enhanced codes was used for universal coding of multiple access quantum channels, where one of the senders shares classical messages with the receiver while the other sends quantum information.\\
	Theorem \ref{avcqmainresult} does not make a positive statement about the structure of $\overline{\cA}_{r,CET}$ in the case where $\overline{\cA}_{d,CET}=\{(0,0)\}$. In \cite{bochenozel}, the authors have constructed an example of a channel where the intersection of  $\overline{\cA}_{r,CET}$ with the $x$-axis is positive and $\overline{\cA}_{d,CET}=\{(0,0)\}$. Future work will consider the structure of the non-zero region in this case along both axes.\\
	The capacity region characterized in Theorem \ref{mainresult} is of a multi-letter nature (requiring a
	limit over many uses of the channel) but might reduce to a single-letter formula for specific cases of compound channels, which is in itself an interesting question to be considered in future work. Ensuing this question, one might suggest  formulas for these capacity regions that offer a more useful characterization. This means that the alternative characterization could entail larger one-shot regions compared to our $\hat{C}(\mathcal{N}_{s},p,\Psi)$. An instance of such a characterization in the case of perfectly known quantum channels exists in \cite{e2} Theorem 5. Therein however, the authors note that their one-shot trapezoids is the same as rectangular regions offered in \cite{devetak}, when one considers the union over all the one-shot, one-state regions. The converse statement for compound channels implies that other such characterizations, must also reduce to ours. Reduction to single-letter formulas is nevertheless an important criterion when comparing different characterizations. \\
	Today, in classical systems, secure communication is obtained by applying cryptographic methods upon available reliable- communication schemes. Security of the resulting protocol, that can hence be separated into two protocols (one responsible for reliability and the other for security), relies on assumptions such as non-feasibility of certain tasks or the limited computational
	capabilities of illegal receivers. For the next generation of classical communication systems, it is expected that different applications (e.g. secure message transmission, broadcasting of common messages and message transmission), are all implemented by physical coding or " physical layer service integration " schemes (see \cite{
		physlayer}). For quantum systems that offer a larger variety of services, \cite{devetak, e2, e3} were  the first papers in this line of research. The present paper develops solutions for different models of
	channel uncertainty that are unavoidable when implementing such integrated services in real-world communication. Following up on the results of \cite{gisbert}, an interesting direction for future work is towards finding the solution to the arbitrarily varying model for multiple access and broadcast channels as a key step in development of quantum networks.
	 
	\section*{Acknowledgments}
	This work was supported by the BMBF via grant 16KIS0118K.
	\appendix
\section{Approximation of compound channels using nets}\label{mosonyicodes}
	
	\begin{definition}\label{netdef}
		A $\tau$-net in $\cC(\cH, \cK)$ is a finite set $\{\cN_{i}\}_{i=1}^{T}$ with the property that for each $\cN \in\cC(\cH, \cK)$ there is at least
		one $i \in \{1,...,T\}$ with $\parallel\cN -\cN_{i}\parallel_{\diamond} <\tau$. 
	\end{definition}
	Existence of $\tau$-nets in $\cC(\cH, \cK)$ is guaranteed by the compactness
	of $\cC(\cH, \cK)$. The next lemma contains an upper bound on the minimal cardinality of $\tau$-nets. 
	\begin{lemma}(see e.g. \cite{boche18} Lemma 7)\label{lemma14}
		For any $\tau \in (0, 1]$, there is a $\tau$-net $\{\cN_{i}\}_{i=1}^{T}$ in $\cC(\cH, \cK)$ with $T \leq 
		(\frac{3}{\tau})^{2(d\cdot d')^{2}}$, where $d = \dim \cH$
		and $d' = \dim \cK$.
	\end{lemma}
Given a net in $\cC(\cH,\cK)$, any compound channel generated by $\cJ\subset\cC(\cH,\cK)$ can be approximated by one of its finite subsets. This is the subject of the following lemma. 
	
	\begin{lemma}\label{net}(see e.g Lemma 13 \cite{boche17})
		Given any compound channel generated by $\cJ\subset\cC(\cH,\cK)$, one can construct a finite set $\cJ_{\tau}$ with the following properties:
		\begin{enumerate}
			\item $\cJ_{\tau}\subset \cJ$,
			\item $|\cJ_{\tau}|\leq(\frac{6}{\tau})^{2(d\cdot d')^{2}}$ with $d,d'$ the dimensions of $\cH,\cK$ respectively and
			\item for all $\cN\in\cJ,\exists\cN'\in\cJ_{\tau}$ such that $\parallel\cN-\cN'\parallel_{\diamond}\leq 2\tau$.
		\end{enumerate}
	\end{lemma}
	\section{Auxiliary results}\label{auxresults}
	In this section we state some results for reader's convenience. 
	\begin{lemma}\label{appendixlemma}(\cite{fidelitydev})
		Let $\Psi,\rho,\sigma\in\cS(\cK)$ and let $\Psi$ be pure. Then
		\begin{equation*}
		F(\Psi,\rho)\geq F(\Psi,\sigma)-\frac{1}{2}\parallel\rho-\sigma\parallel_{1}
		\end{equation*}
	\end{lemma}
	\begin{lemma}(\cite{boche18})\label{entriwise}
		Let $L$ and $D$ be $N\times N$ matrices with non-negative entries which satisfy
		\begin{equation*}
		L_{jl}\leq L_{jj}, L_{jl}\leq L_{ll}
		\end{equation*}
		and 
		\begin{equation*}
		D_{jl}\leq\max\{D_{jj},D_{ll}\}
		\end{equation*}
		for all $j,l\in\{1,...,N\}$. Then
		\begin{equation*}
		\sum_{j,l=1}^{N}\frac{1}{N}\sqrt{L_{jl}D_{jl}}\leq 2\sum_{j=1}^{N}\sqrt{L_{jj}D_{jj}}.
		\end{equation*}
	\end{lemma}
	
	\begin{lemma}\label{projectionlemma}(\cite{boche17} Lemma 3)
Let $\rho\in\cS(\cH)$ for some Hilbert space $\cH$. Let, for some other Hilbert space $\cK$, $\cA\in\cC(\cH,\cK)$, $\cD\in\cC(\cK,\cH)$, $q\in\cL(\cK)$ be an orthogonal projection. If for some $\epsilon>0$, $F_{e}(\rho,\cD\circ Q\circ\cA)\geq 1-\epsilon$ holds, then we have 
\begin{equation*}
F_{e}(\rho,\cD\circ\cA)\geq 1-3\epsilon.
\end{equation*}
	\end{lemma}
	\begin{lemma}\label{5boche}(see e.g. \cite{boche17} Lemma 5)
		There is a real number $\bar{c}>0$ such that for every Hilbert space $\cH$, there exist functions $h':\mathbb{N}\to\mathbb{R}^{+}$, $\phi:(0,1/2)\to\mathbb{R}^{+}$ with $\lim_{l\to\infty}h'(l)=0$ and $\lim_{\delta\to 0}\phi(\delta)=0$, such that for $\rho\in\cS(\cH)$, $\delta\in(0,1/2)$, $l\in\mathbb{N}$, there is an orthogonal projection $q_{\delta,l}$ called the frequency typical projection satisfying
		\begin{enumerate}
			\item $\tr(\rho^{\otimes l}q_{\delta,l})\geq 1-2^{-l(\bar{c}\delta^{2}-h'(l))}$
			\item $q_{\delta,l}\rho^{\otimes l}q_{\delta,l}\leq 2^{-(S(\rho^{\otimes l})-l\phi(\delta))}q_{\delta,l}$.
		\end{enumerate}
	\end{lemma}
	\begin{lemma}(see e.g. \cite{boche17} Lemma 6)\label{6boche}
		Let $\mathcal{H}$ and $\mathcal{K}$ be finite dimensional Hilbert spaces. There are functions $\gamma:(0,1/2)\to\mathbb{R}^{+}$ and $h':\mathbb{N}\to\mathbb{R}^{+}$ satisfying $\lim_{\delta\to 0}\gamma(\delta)=0$ and $h'(l)\searrow 0$, such that for each $\mathcal{N}\in \cC(\mathcal{H},\mathcal{K})$, $\delta\in (0,1/2)$, $l\in\mathbb{N}$ and maximally mixed state $\pi_{\cG}$ on some $\mathcal{G}\subset\mathcal{H}$, there is an operation $\mathcal{N}_{\delta,l}\in \cC^{\downarrow}(\mathcal{H}^{\otimes l},\mathcal{K}^{\otimes l})$, called the reduced operation with respect to $\mathcal{N}$ and $\pi_{\cG}$, satisfying
		\begin{enumerate}
			\item $\text{tr}(\mathcal{N}_{\delta,l}(\pi_{\cG}^{\otimes l}))\geq 1-2^{-l(c'\delta^{2}-h'(l))}$, with universal constant $c'>0$.
			\item $\mathcal{N}_{\delta,l}$ has a Kraus representation with at most $n_{\delta,l}\leq 2^{S_{e}(\pi_{\cG}^{\otimes l},\mathcal{N}^{\otimes l})+l(\gamma(\delta)+h'(l))}$ Kraus operators.
			\item For every state $\rho\in\cS(\mathcal{H}^{\otimes l})$ and every two channels $\mathcal{M}\in \cC^{\downarrow}(\mathcal{H}^{\otimes l},\mathcal{H}^{\otimes l})$ and $\mathcal{L}\in \cC^{\downarrow}(\mathcal{K}^{\otimes l},\mathcal{H}^{\otimes l})$, the inequality
			\begin{equation*}
			F_{e}(\rho,\mathcal{L}\circ\mathcal{N}_{\delta,l}\circ\mathcal{M})\leq F_{e}(\rho,\mathcal{L}\circ\mathcal{N}^{\otimes l}\circ\mathcal{M})
			\end{equation*}
			is fulfilled.
		\end{enumerate}
	\end{lemma}
	\begin{lemma}[Gentle measurement]\label{gmeasurement} (see e.g.\cite{wilde13})
		Let $\rho\in\cS(\cH)$ and $0\leq\Lambda\leq\mathbbm{I}$ with
		\begin{equation*}
		\tr(\Lambda\rho)\geq 1-\epsilon
		\end{equation*}
		for some $0\leq\epsilon< 1$. Then for $\rho':=\frac{\sqrt{\Lambda}\rho\sqrt{\Lambda}}{\tr(\Lambda\rho)}$ we have
		\begin{equation*}
		\parallel\rho-\rho'\parallel_{1}\leq 2\sqrt{\epsilon}.
		\end{equation*}
	\end{lemma}
		\begin{lemma}\label{decouplemma}(\cite{boche18} proof of Theorem 3.2 equation (16) )
			Let $\mathcal{F}\subset\mathcal{G}\subset\mathcal{H}$ with $\dim(\cF)=k$ be given. Also let any member of the set $\{\mathcal{N}_{1},\dots,\cN_{|S|}\}\subset\cC^{\downarrow}(\mathcal{H},\mathcal{K})$ have a Kraus representation with $n_{j}$ operators for $j\in\{1,\dots,S\}$ and set
			\begin{equation*}
			\overline{\cN}:=\frac{1}{|S|}\sum_{j=1}^{|S|}\cN_{j}.
			\end{equation*} 
			Then there exists a recovery operation $\cR\in\cC(\cK,\cH)$  such that
			\begin{equation}\label{lb}
			F_{e}(\pi_{\mathcal{F}},\mathcal{R}\circ\overline{\mathcal{N}})\geq w-\parallel D(p)\parallel_{1},
			\end{equation}
			where $w:=\text{tr}(\overline{\mathcal{N}}(\pi_{\mathcal{F}}))$,
			$p:=k\pi_{\mathcal{F}}$
			and
			\begin{equation*}
			D(p):=\sum_{j,l=1}^{|S|}\frac{1}{|S|}\sum_{i,r=1}^{n_{j},n_{l}}D_{(ij)(rl)}(p)\otimes\ket{e_{i}}\bra{e_{r}}\otimes\ket{f_{j}}\bra{f_{l}}.
			\end{equation*}
			
			In the above
			\begin{equation*}
			D_{(ij)(rl)}(p):=\frac{1}{k}(p a_{j,i}a^{\dagger}_{l,r}p-\frac{1}{k}tr(pa_{j,i}^{\dagger}a_{l,r}p)p).
			\end{equation*}
			
			where $\{\ket{f_{1}},...,\ket{f_{|S|}}\}$ and $\{\ket{e_{1}},...,\ket{e_{n|S|}}\}$ are arbitrary orthonormal bases for $\mathbb{C}^{|S|}$ and $\mathbb{C}^{n|S|}$, and where $\{a_{j,i}\}_{i=1}^{n_{j}}$ is the set of Kraus operators for $\mathcal{N}_{j}$.
		\end{lemma}
		\begin{lemma}\label{robustlemma}(see \cite{ahls})
			If a function $f:S^{l}\to[0,1]$, satisfies 
			\begin{equation}
			\sum_{s^{l}\in S^{l}}f(s^{l})q^{l}(s^{l})\geq 1-\gamma
			\end{equation}
			with $q^{l}(s^{l}):=\prod_{i=1}^{l}q(s_{i})$, for all $q\in\cT(l,S)$ and some $\gamma\in[0,1]$, then
			\begin{equation}
			\frac{1}{l!}\sum_{\sigma\in \mathfrak{S}_{l}}f(\sigma(s^{l}))\geq 1-(l+1)^{|S|}\cdot\gamma\ \ \forall{s^{l}\in S^{l}}. 
			\end{equation}
		\end{lemma}
		\begin{lemma}\label{2epsilonlemma}(see \cite{ahls})
Let $K\in\mathbb{N}$ and numbers $a_{1},\dots,a_{K},b_{1},\dots,b_{K}\in[0,1]$ be given. Assume that
\begin{equation*}
\frac{1}{K}\sum_{i=1}^{K}a_{i}\geq 1-\epsilon
\end{equation*}
and
\begin{equation*}
\frac{1}{K}\sum_{i=1}^{K}b_{i}\geq 1-\epsilon
\end{equation*}
hold. Then
\begin{equation}
\frac{1}{K}\sum_{i=1}^{K}a_{i}b_{i}\geq 1-2\epsilon.
\end{equation}
		\end{lemma}
	
\end{document}